\documentclass[11pt,a4paper]{article}
\usepackage{amsfonts}
\usepackage{amsmath}
\usepackage{amsbsy}
\usepackage{theorem}
\usepackage{graphicx}

\usepackage{mathtools}
\usepackage{multirow,array}
\usepackage{graphicx}
\usepackage{subfigure}
\usepackage{epstopdf}
\usepackage{float} 
\usepackage{color, xcolor}

\newtheorem{proposition}{Proposition}
\newtheorem{lemma}{Lemma}

\newtheorem{remark}{Remark}
\allowdisplaybreaks

\begin{document}

\sf
\title{\bf Climate change analysis from LRD manifold functional regression
}
\date{}

 \maketitle


\author{Diana  P. Ovalle--Mu\~noz$^{1},$ M. Dolores Ruiz--Medina$^{1}$}

\noindent{\small $^{1}$\ Department of Statistics and Operation Research, University of Granada
}

\begin{abstract}This work is motivated by the problem of predicting  downward solar radiation flux spherical maps from the observation of
atmospheric pressure at high cloud bottom. To this aim  nonlinear functional regression is implemented under strong--correlated functional data. The link operator  reflects   the heat transfer    in the atmosphere.  A latent  parametric  linear  functional regression model reduces uncertainty in the support of this operator.  An additive  long--memory manifold--supported functional time series error  models persistence in time of random fluctuations observed in the response. Time is incorporated via the scalar covariates in the latent  linear  functional regression model. The functional  regression parameters in this model are supported on a connected and compact two point homogeneous space. Its
 Generalized Least--Squares (GLS) parameter estimation   is achieved. When  the second--order structure of the functional error term  is unknown, its minimum contrast estimation is obtained in the spectral domain. The performance  of the theoretical and plug--in  nonlinear functional  regression predictors   is illustrated in the  simulation study undertaken in the sphere.  The  Supplementary Material provides a detailed empirical analysis in the one way ANOVA context.
   The real--data application extends the   purely spatial statistical analysis of atmospheric pressure at high cloud bottom, and downward solar radiation flux in \cite{Alegria21} to the spatiotemporal context.

\end{abstract}


\maketitle

\section{Introduction}
\label{sec:1}

Solar radiation has experimented important intensity and distributional changes    in the last few decades affecting  climate and global warming.  Well--known factors  are, for instance,  greenhouse gases,  trapping  more heat in the atmosphere, and increasing the Earth's surface temperature.  Atmospheric aerosols (e.g.,  smog or pollution particles) can scatter solar radiation,  and  locally modify the  energy reaching the Earth's surface.  The  Hydrological Cycle is also affected  by the evaporation caused by high  temperatures produced by  solar radiation   on  the surface, affecting atmospheric pressure and precipitation.
Thus, the interaction between distribution and intensity of solar radiation, and atmospheric pressure seriously affects regional and global climate systems (see, e.g., \cite{Forster}; \cite{Wild09}). It is well-known that solar radiation and atmospheric pressure are interconnected through thermodynamic, radiative, and fluid dynamics processes involved in   the formulation of  associated physical equations like Radiative Transfer, Hydrostatic, Ideal Gas, Energy Balance, and Atmospherical  Heat Transfer equations.

The   functional regression  approach proposed in this work is motivated by the above    physical  research subarea.  The physical law governing the  coupled dynamic of atmospheric pressure at high cloud bottom, and downward solar radiation flux is incorporated in the definition of  the link nonlinear operator, modelling heat transfer    in the atmosphere.  Restricting our attention to this thermodynamics law means that we are  ignoring the physical information supported by the remaining  physical equations, introducing uncertainty in our model, reflected in the  observed random   fluctuations of the  response  given by the solar radiation flux. In our modelling framework,  a  latent parametric  functional linear  regression  model is considered to reduce   uncertainty in  the support of our  link nonlinear  operator.   In this modelling framework, we also incorporate an additive   strong--dependent functional time series error term representing the   structured   small--scale random  fluctuations in space and time.

The adopted   $H$--valued time series framework  refers to the separable Hilbert space  $H=L^{2}\left(\mathbb{M}_{d},d\nu,\mathbb{R}\right)$ of real--valued square integrable functions with compact
 support contained in the manifold
 $\mathbb{M}_{d}.$ Here,  $d\nu$ denotes  the normalized Riemannian measure   on $\mathbb{M}_{d}.$
 Our functional  regression modelling framework goes beyond the structural assumptions present, for instance,  in   \cite{ALVRuiz19}; \cite{Cuevas02};  \cite{FebreroBande15}; \cite{Bosq2000};    \cite{Cardot07}; \cite{Crambes13};  \cite{Ma99}; \cite{RuizMedina11}; \cite{RuizMedina12a};  \cite{RuizMedina12b},  and  references therein.
It also supposes an extension to the nonlinear,  non Euclidean,  and Long--Range Dependence (LRD) settings of the FANOVA analysis under weak--dependent errors
 achieved in \cite{RuizMedina16}, and later, in  \cite{Alvarez17}  including the case of circular domains (see also \cite{RuizMD18} in the context of  multiple functional  regression).

A purely spatial statistical analysis of atmospheric pressure at high cloud bottom, and downward solar radiation flux has been achieved by \cite{Alegria21}, in the framework of  spherical  isotropic  random fields from a nonparametric bayesian perspective. As pointed out in \cite{Alegria21},  their purely spatial bivariate   analysis of    atmospheric pressure, and downward solar radiation flux data involves some information loss on the  temporal patterns in these data due to the averaging performed over time.
The present paper addresses this problem incorporating the  time dynamics  into the covariates defining  the  latent parametric  functional linear  regression  model, and  a  regression manifold--supported functional error term strong--correlated in time. Thus, prediction of  downward solar radiation flux  earth map sequences is performed, conditioning to the    atmospheric pressure at high
cloud bottom, in the context of  functional nonlinear regression.

Our  approach   also  supposes a substantial contribution to the field of spatiotemporal regression from a functional perspective. Note that
 most of the functional regression tools have been developed under the assumption of independent or weak--dependent functional data
(see, e.g.,  \cite{Caponera21};  \cite{CaponeraMarinucci};  \cite{Panaretos13}  in the functional time series framework). GLS estimation of the latent functional  regression parameter is implemented in our approach
to incorporate the  strong--dependence structure of the error term in the estimation procedure.
 The computational cost and complexity of the  implementation can be substantially reduced, under the invariance of the covariance  kernels with respect to the group of isometries of the connected and compact two--point homogeneous space $\mathbb{M}_{d}$  with topological dimension   $d.$  Invariant kernels admit a diagonal series expansion   in terms of the eigenfunctions of the Laplace--Beltrami operator.
We adopt the approach introduced in  \cite{RuizMedina2022}  in our LRD spectral  analysis of the regression functional error term, when its spatiotemporal dependence structure is unknown.
Specifically, minimum contrast estimation in the spectral domain is implemented  (see  also \cite{Ovalle23}  in the framework of  multifractional integration of spherical functional time series).
   Note that the LRD analysis of functional time series in the nonstationary case has mainly been developed in terms of the eigendecomposition of the long run covariance function (see, e.g., \cite{LiRobinsonShang19}).

The  outline of the paper is as follows. Section \ref{secprel}  presents some preliminary elements on the spectral analysis of LRD manifold--suported  functional time series. Section \ref{secFANOVA} introduces the regression model and estimation methodology. In Section \ref{secsimulation}, a simulation study is undertaken to illustrate the performance of the theoretical and plug--in nonlinear regression predictors, under an infinite--dimensional log--Gaussian scenario.  See also Section 1 in the  Supplementary Material  where  GLS estimation under strong--dependent  data is illustrated in the context of  spherical functional linear regression models.  The asymptotic and finite functional sample size properties of the corresponding theoretical and plug--in regression predictors are displayed as well. Particularly,  the effect of the pure point spectral properties of the LRD operator   of the  regression error, affecting  accuracy and variability of the GLS plug--in parameter estimator,  is analyzed.
Section \ref{secApplic} implements the proposed functional theoretical and plug--in nonlinear regression predictors from the generated synthetic data of downward solar radiation and  atmospheric pressure at high cloud bottom. The performance of the functional regression predictor is evaluated in terms of $5$--fold random cross validation.
Some final comments and open research lines are discussed in Section \ref{secconclusions}.

\section{Preliminaries}

\label{secprel}

 Let   $X=\{ X(\mathbf{x},t),\ \mathbf{x}\in \mathbb{M}_{d},\ t\in \mathbb{T}\}$  be  a zero--mean, stationary in time, and isotropic  in space mean--square continuous Gaussian, or elliptically contoured, spatiotemporal random field  on the basic probability space $(\Omega,\mathcal{A},P),$  with covariance function
$C(d_{\mathbb{M}_{d}}(\mathbf{x},\mathbf{y}),t-s)= E\left[X(\mathbf{x},t) X(\mathbf{y},s)\right],$ for $\mathbf{x},\mathbf{y}\in \mathbb{M}_{d},$ and $t,s\in \mathbb{T}.$ Here, $\mathbb{T}$ denotes the temporal domain, which can be $\mathbb{Z}$ or $\mathbb{R}.$ Under the conditions of   Theorem  4   in   \cite{MaMalyarenko}, the covariance function $C(d_{\mathbb{M}_{d}}(\mathbf{x},\mathbf{y}),t-s)$ admits the following diagonal series expansion:
\begin{eqnarray}&&C(d_{\mathbb{M}_{d}}(\mathbf{x},\mathbf{y}),t-s)=
  \sum_{n\in \mathbb{N}_{0}}
B_{n}(t-s)\sum_{j=1}^{\delta (n,d)}S_{n,j}^{d}(\mathbf{x})S_{n,j}^{d}(\mathbf{y})\nonumber\\
&&=
\sum_{n\in \mathbb{N}_{0}} \frac{\delta (n,d)}{\omega_{d}}  B_{n}(t-s)R_{n}^{(\alpha, \beta )}\left(\cos\left(d_{\mathbb{M}_{d}}(\mathbf{x},\mathbf{y})\right)\right),\ \mathbf{x},\mathbf{y}\in \mathbb{M}_{d}, \ t,s \in \mathbb{T},\nonumber\\ \label{klexpc2}
\end{eqnarray}
\noindent where $\delta (n,d)$ denotes the dimension of the $n$th eigenspace $\mathcal{H}_{n}$ of the Laplace Beltrami operator,
$\omega_{d}=\int_{\mathbb{M}_{d}}d\nu(\mathbf{x}),$
and
  $\{ S_{n,j}^{d},\ j=1,\dots,\delta(n,d), \ n\in \mathbb{N}_{0}\}$   is the  system of  orthonormal  eigenfunctions
of the Laplace Beltrami operator $\Delta_{d}$ on $L^{2}(\mathbb{M}_{d},d\nu ,\mathbb{R}).$
 Furthermore, in the last identity in (\ref{klexpc2}), we have applied addition formula  in the context of connected and compact two--point homogeneous spaces
 (see Theorem 3.2. in \cite{Gine1975} and p. 455 in \cite{Andrews99}), where
$R_{n}^{\alpha,\beta}\left(\cos(d_{\mathbb{M}_{d}}(\mathbf{x},\mathbf{y}))\right)
=\frac{P_{n}^{\alpha,\beta}\left(\cos(d_{\mathbb{M}_{d}}(\mathbf{x},\mathbf{y}))\right)}{P_{n}^{\alpha,\beta}\left(1\right)},$ with $P_{n}^{\alpha,\beta}$ denoting the Jacobi polynomial of degree $n \in\mathbb{N}_{0},$ with parameters $\alpha$ and $\beta$ (see, e.g., \cite{MaMalyarenko}, and \cite{Cartan1927}, for more details on Lie Algebra based  approach).

Consider the restriction   $X_{T}=\{ X(\mathbf{x},t),\ \mathbf{x}\in \mathbb{M}_{d},\ t\in [0,T]\}$ of $X$  satisfying (\ref{klexpc2})
 to the interval $[0,T].$  The following lemma provides the orthogonal expansion of $X_{T}=\{ X(\mathbf{x},t),\ \mathbf{x}\in \mathbb{M}_{d},\ t\in [0,T]\}$ in terms of the eigenfunctions of the Laplace Beltrami operator (see Theorem 1 in the Supplementary Material in \cite{Ovalle23}).
\begin{lemma}
\label{lemmat5ma}
Let $X_{T}=\{ X(\mathbf{x},t),\ \mathbf{x}\in \mathbb{M}_{d},\ t\in [0,T]\}$  be the restriction of $X$  to the interval $[0,T],$ satisfying (\ref{klexpc2}),
and
\begin{equation}\sum_{n\in \mathbb{N}_{0}}B_{n}(0)\delta (n,d)<\infty.\label{eqtrace}
\end{equation}
\noindent Then,  $X_{T}$ admits the following orthogonal expansion:

 \begin{eqnarray}
&&X_{T}(\mathbf{x},t)\underset{\mathcal{L}^{2}_{\widetilde{H}}(\Omega,\mathcal{A},P)}{=}\sum_{n\in \mathbb{N}_{0}} \sum_{j=1}^{\delta (n,d)}V_{n,j}(t)S_{n,j}^{d}(\mathbf{x}),\quad \mathbf{x}\in \mathbb{M}_{d},\ t\in [0,T],
\label{klexp}
\end{eqnarray}

\noindent where $\mathcal{L}^{2}_{\widetilde{H}}(\Omega,\mathcal{A},P)=L^{2}(\Omega\times\mathbb{M}_{d}\times[0,T],P(d\omega )\otimes d\nu\otimes dt),$ with
  $\widetilde{H}=L^{2}(\mathbb{M}_{d}\times[0,T],
 d\nu\otimes dt).$ Here,
 $\{V_{n,j}(t),\ t\in  [0,T], \ j=1,\dots, \delta(n,d), \  n\in \mathbb{N}_{0}\}$ is a sequence of centered  uncorrelated  random processes on $[0,T]$ satisfying
 \begin{equation}
 V_{n,j}(t)=\int_{\mathbb{M}_{d}}X_{T}(\mathbf{y},t)S_{n,j}^{d}(\mathbf{y})d\nu(\mathbf{y}), \  j=1,\dots, \delta(n,d),\ n\in \mathbb{N}_{0},
 \label{fcA}
 \end{equation}
 \noindent and
 \begin{equation}
  E[V_{n,j}(t)V_{m,l}(s)]=\delta_{n,m}\delta_{j,l}B_{n}(t-s)
  \label{eqincorr}
  \end{equation}
\end{lemma}

Assume  that $\mathbb{T}=\mathbb{Z},$ and that the map $$\widetilde{X}_{t}:(\Omega ,\mathcal{A})\longrightarrow \left(L^{2}(\mathbb{M}_{d},d\nu ,\mathbb{R}),\mathcal{B}(L^{2}(\mathbb{M}_{d},d\nu ,\mathbb{R}))\right)$$ \noindent  is   measurable, with  $\widetilde{X}_{t}(\mathbf{x}):=X(\mathbf{x},t)$ for every $t\in \mathbb{T}$ and $\mathbf{x}\in \mathbb{M}_{d}.$  Here,  $\mathcal{B}(L^{2}(\mathbb{M}_{d},d\nu ,\mathbb{R})) $ denotes the Borel  $\sigma$--algebra of $L^{2}(\mathbb{M}_{d},d\nu, \mathbb{R})$
(i.e., the smallest $\sigma$--algebra containing  the collection of all open subsets of  $L^{2}(\mathbb{M}_{d},d\nu, \mathbb{R})$).
By previous assumptions on $X,$  $\left\{\widetilde{X}_{t}, \ t\in \mathbb{Z}\right\}$   then defines a manifold--supported  weak--sense stationary functional time series. In particular, $E\left[\widetilde{X}_{t}\right]=0,$ and $\sigma_{\widetilde{X}}^{2}=E\left[\|\widetilde{X}_{t}\|^{2}_{L^{2}(\mathbb{M}_{d},d\nu,\mathbb{R})}\right]=
E\left[\|\widetilde{X}_{0}\|^{2}_{L^{2}(\mathbb{M}_{d},d\nu,\mathbb{R})}\right]=\|R_{0}\|_{L^{1}(L^{2}(\mathbb{M}_{d},d\nu,\mathbb{R}))},$ for every $t\in \mathbb{Z}.$ By $L^{1}(L^{2}(\mathbb{M}_{d},d\nu,\mathbb{R}))$ we denote  the space of trace or nuclear operators on $L^{2}(\mathbb{M}_{d},d\nu,\mathbb{R}).$
The second--order structure of $\left\{\widetilde{X}_{t}, \ t\in \mathbb{Z}\right\}$ is  characterized by the family of covariance operators  $\left\{\mathcal{R}_{t}, \ t\in \mathbb{Z}\right\}$  given by, for all $h,g\in L^{2}(\mathbb{M}_{d},d\nu,\mathbb{R}),$
\begin{eqnarray}
&&
\mathcal{R}_{t}(g)(h)= E[\widetilde{X}_{s+t}(h)\widetilde{X}_{s}(g)]= E\left[\left\langle
\widetilde{X}_{s+t},h\right\rangle_{L^{2}(\mathbb{M}_{d},d\nu,\mathbb{R})}\left\langle \widetilde{X}_{s},g\right\rangle_{L^{2}(\mathbb{M}_{d},d\nu,\mathbb{R})}
\right]\nonumber\\
&&
\mathcal{R}_{t}:=E[\widetilde{X}_{s+t}\otimes \widetilde{X}_{s}]=\quad  
\forall t,s\in \mathbb{Z}.
\label{covstatfunct2}
\end{eqnarray}

 Under (\ref{klexpc2}), the family of  covariance operators $\left\{\mathcal{R}_{t},\ t\in \mathbb{Z}\right\}$  satisfies
\begin{equation}\mathcal{R}_{t}=E\left[\widetilde{X}_{t}\otimes \widetilde{X}_{0}\right]= \sum_{n\in \mathbb{N}_{0}}
B_{n}(t)\sum_{j=1}^{\delta (n,d)}S_{n,j}^{d}\otimes S_{n,j}^{d},\ t\in \mathbb{Z}.
\label{sdo3}
\end{equation}

The  spectral density operator family     $\left\{\mathcal{F}_{\omega},\ \omega \in [-\pi, \pi]\right\}$ is   defined from the functional Fourier transform of the elements of the covariance operator family.

  \begin{equation}
  \mathcal{F}_{\omega}
  \underset{\mathcal{S}(L^{2}(\mathbb{M}_{d},d\nu,\mathbb{C}))}{=}
  \frac{1}{2\pi} \sum_{t\in \mathbb{Z}}\exp\left(-i\omega t\right)
\mathcal{R}_{t},\ \omega \in [-\pi,\pi]\backslash \{0\},
\label{sdo2}
\end{equation}
\noindent where $\underset{\mathcal{S}(L^{2}(\mathbb{M}_{d},d\nu,\mathbb{C}))}{=}$ denotes the identity in the norm of the space of Hilbert--Schmidt operators. From equations (\ref{sdo3}) and (\ref{sdo2}),
\begin{eqnarray} &&
  \mathcal{F}_{\omega}
  \underset{\mathcal{S}(L^{2}(\mathbb{M}_{d},d\nu,\mathbb{C}))}{=}\sum_{n\in \mathbb{N}_{0}}\left[\sum_{t\in \mathbb{Z}} \exp\left(-i\omega t\right)B_{n}(t)\right]\sum_{j=1}^{\delta (n,d)}S_{n,j}^{d}\otimes S_{n,j}^{d}\nonumber\\
  &&\underset{\mathcal{S}(L^{2}(\mathbb{M}_{d},d\nu,\mathbb{C}))}{=}\sum_{n\in \mathbb{N}_{0}}f_{n}(\omega )\sum_{j=1}^{\delta (n,d)}   S_{n,j}^{d}\otimes S_{n,j}^{d},\nonumber
\label{othexp}
\end{eqnarray}
\noindent with
\begin{equation}B_{n}(t)= \int_{[-\pi,\pi]} \exp\left(i\omega t  \right) f_{n}(\omega )d\omega ,\quad  \forall t\in \mathbb{Z}.\label{eq1}
\end{equation}

 The functional Discrete Fourier Transform fDFT $\widetilde{X}^{(T)}_{\omega }(\cdot)$ of the map  data sequence  is defined as
\begin{equation}\widetilde{X}^{(T)}_{\omega }(\mathbf{x})
=\frac{1}{\sqrt{2\pi
T}}\sum_{t=1}^{T}\widetilde{X}_{t}(\mathbf{x})\exp\left(-i\omega t\right), \quad \mathbf{x}\in \mathbb{M}_{d},\
\omega\in [-\pi ,\pi].\label{fDFT}\end{equation}
\noindent Note that $$E\left[\|\widetilde{X}_{\omega}^{(T)}\|_{L^{2}(\mathbb{M}_{d},d\nu,\mathbb{C})}\right]\leq \frac{1}{\sqrt{2\pi
T}}\sum_{t=1}^{T}E\|\widetilde{X}_{t}(\cdot )\|_{L^{2}(\mathbb{M}_{d},d\nu,\mathbb{R})}<\infty.$$
\noindent Then,
 $\widetilde{X}^{(T)}_{\omega }(\cdot)$ is a random  element in the space   $L^{2}(\mathbb{M}_{d},d\nu,\mathbb{C}),$ which denotes  the complex version of the Hilbert space  $L^{2}(\mathbb{M}_{d},d\nu,\mathbb{R}).$

As usually,  the periodogram operator is defined  from the fDFT by $p_{\omega }^{(T)}=\widetilde{X}_{\omega}^{(T)}\otimes
\overline{\widetilde{X}_{\omega}^{(T)}}=\widetilde{X}_{\omega}^{(T)}
 \otimes \widetilde{X}_{-\omega}^{(T)}.$   Its mean is then computed as
  \begin{eqnarray}E[p_{\omega }^{(T)}]&=&E[\widetilde{X}_{\omega}^{(T)}
 \otimes \widetilde{X}_{-\omega}^{(T)}]=  \frac{1}{2\pi }\sum_{u=-(T-1)}^{T-1}\exp\left(-i\omega u\right)
\frac{(T-|u|)}{T}\mathcal{R}_{u}\nonumber \\&=&\int_{-\pi}^{\pi} F_{T}(\omega - \xi)
\mathcal{F}_{\xi} d\xi,\quad T\geq 2,\nonumber \end{eqnarray}
\noindent in terms of the F\'ejer kernel $F_{T}(\omega )=\frac{1}{T}\sum_{t=1}^{T}
\sum_{s=1}^{T}\exp\left(-i(t-s)\omega \right).$

\section{Multiple  functional regression in manifolds}
\label{secFANOVA}

Let us consider the functional regression model
\begin{eqnarray}
&&\mathbf{Y}(\mathbf{y})=\mathbf{H}\left(\mathbf{X}(\boldsymbol{\beta})\right)(\mathbf{y})+\boldsymbol{\varepsilon}(\mathbf{y}),\quad
\mathbf{y}\in \mathbb{M}_{d},
\label{ecmodeloch5}
\end{eqnarray}
\noindent given by
\begin{eqnarray}
&&
\left[\begin{array}{c} Y_{1}(\mathbf{y})\\
 \vdots\\
 Y_{T}(\mathbf{y})\\
 \end{array}\right]=H\left(\left[\begin{array}{c}
 \sum_{j=1}^{p}X_{1,j}\beta_{j}(\mathbf{y})\\
 \vdots\\
 \sum_{j=1}^{p}X_{T,j}\beta_{j}(\mathbf{y})\\
 \end{array}\right]\right)+\left[\begin{array}{c} \varepsilon_{1}(\mathbf{y})\\
 \vdots \\
 \varepsilon_{T}(\mathbf{y})\\
 \end{array}
 \right]\nonumber\\
 &&=H\left(\left[\begin{array}{c}
 g_{1}(\underline{X}_{1},\beta (\mathbf{y}))\\
 \vdots\\
 g_{T}(\underline{X}_{T},\beta (\mathbf{y}))
 \end{array}\right]\right)+\left[\begin{array}{c} \varepsilon_{1}(\mathbf{y})\\
 \vdots \\
 \varepsilon_{T}(\mathbf{y})\\
 \end{array}
 \right],\ \mathbf{y}\in \mathbb{M}_{d},\nonumber
\end{eqnarray}
\noindent where
\begin{eqnarray} &&\mathbf{X}=(X_{t,j})_{t=1,\dots,T; j=1,\dots, p}; \quad  \mathbf{Y}(\mathbf{y})=[Y_{1}(\mathbf{y}),Y_{2}(\mathbf{y}),\dots,Y_{T}(\mathbf{y})]^{T},\ \mathbf{y}\in \mathbb{M}_{d}\nonumber\\
&&\underline{X}_{t}=[X_{t,1},\dots, X_{t,p}]^{T},\quad   t=1,\dots,T,\nonumber\\
&&\boldsymbol{\beta}(\mathbf{y})=[\beta_{1}(\mathbf{y}),\dots, \beta_{p}(\mathbf{y}) ]^{T};\quad
 \boldsymbol{\varepsilon}(\mathbf{y})=[\varepsilon_{1}(\mathbf{y}),\varepsilon_{2}(\mathbf{y}),\dots,\varepsilon_{T}(\mathbf{y})]^{T},\
 \mathbf{y}\in \mathbb{M}_{d}.\nonumber
 \end{eqnarray}
 \noindent  In equation (\ref{ecmodeloch5}), the regression
 parameters   $\beta_{j}\in L^{2}(\mathbb{M}_{d},d\nu ,\mathbb{R}),$ $j=1,\dots,p,$ respectively  provide the spatial weighting of the time--varying  covariates \linebreak $(X_{t,j}\in \mathbb{R},\ t=1,\dots,T),$ $j=1,\dots,p.$ The measurable
 mapping \linebreak
   $\mathbf{H}:[L^{2}(\mathbb{M}_{d},d\nu ,\mathbb{R})]^{T}\to [L^{2}(\mathbb{M}_{d},d\nu ,\mathbb{R})]^{T}$  is assumed to be an   isomorphic (bijective and bicontinuous). This mapping   combines geographical and temporal information affecting the functional response $\mathbf{Y}$ in a nonlinear manner.
   In particular, in our context,  it is defined from the physical law governing the dynamical relationship between $\mathbf{Y}(\cdot)$ and $\mathbf{X}.$ Hence, $\mathbf{H} $  is assumed to be known. Here, $[L^{2}(\mathbb{M}_{d},d\nu ,\mathbb{R})]^{T}$ denotes the separable Hilbert space
of $T$--dimensional vector functions with the inner product \begin{equation}\left\langle\mathbf{f},\mathbf{g}\right\rangle_{[L^{2}(\mathbb{M}_{d},d\nu,\mathbb{R})]^{T}}=\sum_{l=1}^{T}\left\langle f_{l},g_{l}\right\rangle_{L^{2}(\mathbb{M}_{d},d\nu,\mathbb{R})},
\label{eqfss}
\end{equation} \noindent  for every $\mathbf{f}=(f_{1},\dots,f_{T})^{T},$ $\mathbf{g}=(g_{1},\dots,g_{T})^{T}\in [L^{2}(\mathbb{M}_{d},d\nu,\mathbb{R})]^{T}.$

 \begin{remark}
In some applied fields,    $(X_{t,j},\ t=1,\dots, T),$ $j=1,\dots,p,$  can represent the observed values of the time--varying   Fourier coefficients of the spatiotemporal covariates, with respect to the spatial basis we want to fit in an optimal least--squares sense, to characterize the functional support of the link operator $\mathbf{H}.$
\end{remark}
 We assume that the error term
$ \boldsymbol{\varepsilon}(\mathbf{y})$ in (\ref {ecmodeloch5}) is independent of $(\underline{X}_{1},\dots ,\underline{X}_{T}),$ for every $\mathbf{y}\in \mathbb{M}_{d}.$
Process  $\{\varepsilon_{t},\ t\in \mathbb{Z}\}$
defines an   LRD  stationary zero--mean functional time series, with values in the space  $L^{2}(\mathbb{M}_{d},d\nu,\mathbb{R}),$ having   invariant covariance operators   with  respect to the group of isometries of $\mathbb{M}_{d},$
satisfying the conditions assumed in Theorem  4   in \cite{MaMalyarenko} and in Lemma \ref{lemmat5ma} in Section \ref{secprel}.

\begin{remark}
 Note that in equation (\ref{ecmodeloch5}), for each $\mathbf{y}\in \mathbb{M}_{d},$
\begin{eqnarray}
 &&\mathbf{Y}(\mathbf{y})=\left[\begin{array}{c} Y_{1}(\mathbf{y})\\
 \vdots\\
 Y_{T}(\mathbf{y})\\
 \end{array}\right]=E\left[\mathbf{Y}(\mathbf{y})/(\underline{X}_{1},\dots ,\underline{X}_{T})\right]
+\left[\begin{array}{c} \varepsilon_{1}(\mathbf{y})\\
 \vdots \\
 \varepsilon_{T}(\mathbf{y})\\
 \end{array}
 \right].\nonumber
 \end{eqnarray}
 \noindent Hence,
\begin{eqnarray}&&  E\left[\mathbf{H}^{-1}(\mathbf{Y})(\mathbf{y})/(\underline{X}_{1},\dots ,\underline{X}_{T})\right]=\mathbf{X}(\boldsymbol{\beta})(\mathbf{y})=
\left[\begin{array}{c}
 g_{1}(\underline{X}_{1},\beta (\mathbf{y}))\\
 \vdots\\
 g_{T}(\underline{X}_{T},\beta (\mathbf{y}))
 \end{array}\right],\ \mathbf{y}\in \mathbb{M}_{d}.\nonumber\\ \label{eqob2}\end{eqnarray}
 \end{remark}
  Assume that  the Fr\'echet Jacobian of $\mathbf{H}^{-1}$ is almost surely bounded. Thus,
    the absolute difference of the probability  distributions of $\mathbf{Y}(\mathbf{y})$ and $\widetilde{\mathbf{Y}}(\mathbf{y})=\mathbf{H}^{-1}(\mathbf{Y})(\mathbf{y})$
  is uniformly bounded  (only depending on the spatial location  $\mathbf{y}\in \mathbb{M}_{d}$). One  can then consider the following regression model:
 \begin{equation}
\widetilde{\mathbf{Y}}(\mathbf{y})= H^{-1}(\mathbf{Y})(\mathbf{y})=\mathbf{X}(\boldsymbol{\beta})(\mathbf{y})+\boldsymbol{\varepsilon}(\mathbf{y}),\quad \mathbf{y}\in \mathbb{M}_{d}.
\label{ellr}
\end{equation}
We refer  to model  (\ref{ellr}) as the latent  parametric linear functional regression model.

The GLS estimator of the  functional  parameter vector $\boldsymbol{\beta}$ is computed in the next section, incorporating the  spatiotemporal dependence structure of the functional  error term  $\boldsymbol{\varepsilon}(\cdot),$ which is
 given by the
 matrix covariance operator
\[
\begin{split}
\mathbf{R}_{\boldsymbol{\varepsilon}\boldsymbol{\varepsilon}}&=
E\left[\boldsymbol{\varepsilon}(\cdot)\boldsymbol{\varepsilon}^{T}(\cdot)\right]\\
 &=\begin{bmatrix}
E\left[\varepsilon_{1}(\cdot)\otimes\varepsilon_{1}(\cdot)\right] & \cdots & E\left[\varepsilon_{1}(\cdot)\otimes\varepsilon_{T}(\cdot)\right]\\
E\left[\varepsilon_{2}(\cdot)\otimes \varepsilon_{1}(\cdot)\right]& \cdots & E\left[\varepsilon_{2}(\cdot)\otimes\varepsilon_{T}(\cdot)\right]\\
\vdots & \vdots & \vdots \\
E\left[\varepsilon_{T}(\cdot)\otimes \varepsilon_{1}(\cdot)\right]  & \cdots & E\left[\varepsilon_{T}(\cdot)\otimes\varepsilon_{T}(\cdot)\right]
\end{bmatrix}\\
&=\begin{bmatrix}
\mathcal{R}_{0} & \mathcal{R}_{1} & \cdots & \mathcal{R}_{T-1}\\
\mathcal{R}_{1} & \mathcal{R}_{0}& \cdots & \mathcal{R}_{T-2}\\
\vdots & \vdots & \vdots \\
\mathcal{R}_{T-1}& \mathcal{R}_{T-2}& \cdots & \mathcal{R}_{0}
\end{bmatrix},
\end{split}
\]
  \noindent  where $\mathcal{R}_{T-t}
  =E\left[\varepsilon_{t}\otimes \varepsilon_{T}\right],$ $t=1,\dots,T.$
Note that the functional entries of $\mathbf{R}_{\boldsymbol{\varepsilon}\boldsymbol\varepsilon}$  admit the diagonal series expansion  introduced in equation (\ref{klexpc2}).
In the subsequent development we will consider the following  orthogonal expansion  of $\mathbb{M}_{d}$--supported    functions
\begin{eqnarray}
\label{ecseriesbeta}
\beta_{h}(\mathbf{y})&=&\sum_{n\in\mathbb{N}_{0}}\sum_{k=1}^{\delta(n,d)}\beta_{n,k}^{(h)}S_{n,k}^{d}(\mathbf{y})\nonumber\\
&=&\sum_{n\in\mathbb{N}_{0}}\sum_{k=1}^{\delta(n,d)}\left\langle  \beta_{h},  S_{n,k}^{d}\right\rangle_{L^{2}(\mathbb{M}_{d},d\nu ,\mathbb{R})}S_{n,k}^{d}(\mathbf{y}),\quad \mathbf{y}\in \mathbb{M}_{d},\ h=1,2,\dots,p.
\end{eqnarray}

Under the conditions assumed on the error term in equation (\ref{ecmodeloch5}) (see Lemma \ref{lemmat5ma}),
 from equations (\ref{ellr}) and (\ref{ecseriesbeta}),  the following vector series expansion holds for the response $\widetilde{\mathbf{Y}}(\mathbf{y})$   in the space $\mathcal{L}^{2}_{[L^{2}(\mathbb{M}_{d},d\nu ,\mathbb{R})]^{T}}(\Omega,\mathcal{A},P),$

\begin{eqnarray}
\label{ecexpansionres}\widetilde{\mathbf{Y}}(\mathbf{y})
&=&\sum_{n\in\mathbb{N}_{0}}\sum_{k=1}^{\delta(n,d)}\left\langle \widetilde{\mathbf{Y}}, S_{n,k}^{d}\right\rangle_{L^{2}(\mathbb{M}_{d},d\nu,\mathbb{R})} S_{n,k}^{d}(\mathbf{y}),\quad \mathbf{y}\in \mathbb{M}_{d},
\end{eqnarray}
\noindent where

\begin{eqnarray}&&\hspace*{-0.7cm}\left( \left\langle \widetilde{\mathbf{Y}}, S_{n,k}^{d}\right\rangle_{L^{2}(\mathbb{M}_{d},d\nu,\mathbb{R})}S_{n,k}^{d}(\mathbf{y})\right)^{T}
\nonumber\\
&&\hspace*{2cm}=
\left(\int_{\mathbb{M}_{d}}\widetilde{Y}_{1}(\mathbf{y})S_{n,k}^{d}(\mathbf{y})d\nu(\mathbf{y}),\dots,
\int_{\mathbb{M}_{d}}\widetilde{Y}_{T}(\mathbf{y})S_{n,k}^{d}(\mathbf{y})d\nu(\mathbf{y}) \right)_{1\times T}
\nonumber\\ &&\hspace*{5cm}\times
\mbox{diag}\left( S_{n,k}^{d}(\mathbf{y}),\dots,S_{n,k}^{d}(\mathbf{y})\right)_{T\times T}\nonumber\\
&&\hspace*{-0.5cm}
=\left(\left\langle \widetilde{Y}_{1},S_{n,k}^{d}\right\rangle_{L^{2}(\mathbb{M}_{d},d\nu,\mathbb{R})},\dots, \left\langle \widetilde{Y}_{T},S_{n,k}^{d}\right\rangle_{L^{2}(\mathbb{M}_{d},d\nu,\mathbb{R})}\right)_{1\times T}\mbox{diag}\left( S_{n,k}^{d}(\mathbf{y}),\dots,S_{n,k}^{d}(\mathbf{y})\right)_{T\times T},\nonumber\\
\label{inversion2}
\end{eqnarray}
\noindent with $\mbox{diag}\left( S_{n,k}^{d}(\mathbf{y}),\dots,S_{n,k}^{d}(\mathbf{y})\right)_{T\times T}$ being a diagonal matrix with constant
entries equal to $S_{n,k}^{d}(\mathbf{y}),$ for each  $\mathbf{y}\in \mathbb{M}_{d}.$ Here, from equations (\ref{ellr}) and  (\ref{ecseriesbeta}), for $k=1,\dots, \delta (n,d),$ and $n\in \mathbb{N}_{0},$
\begin{eqnarray}&&\widetilde{\mathbf{Y}}_{n,k}=\left[\begin{array}{c}\left\langle\widetilde{Y}_{1},S_{n,k}^{d}\right\rangle_{L^{2}(\mathbb{M}_{d},d\nu,\mathbb{R})}\\
\vdots\\
  \left\langle\widetilde{Y}_{T},S_{n,k}^{d}\right\rangle_{L^{2}(\mathbb{M}_{d},d\nu,\mathbb{R})}\\
  \end{array}\right]=\left[\begin{array}{c}\widetilde{Y}_{n,k}(1)\\
\vdots\\
  \widetilde{Y}_{n,k}(T)\\
  \end{array}\right]=\left[\begin{array}{c}
 \sum_{j=1}^{p}X_{1,j}\beta_{n,k}^{(j)}+\varepsilon_{n,k}(1)\\
 \vdots\\
 \sum_{j=1}^{p}X_{T,j}\beta_{n,k}^{(j)}+\varepsilon_{n,k}(T)\\
 \end{array}\right]\nonumber\\
 \label{fclr}
\end{eqnarray}
\noindent where
\begin{eqnarray}&&\boldsymbol{\beta}_{n,k}=\left(\beta_{n,k}^{(1)},\dots,\beta_{n,k}^{(p)}\right)^{T}\nonumber\\
&&=\left(\left\langle \beta_{1},S_{n,k}^{d}\right\rangle_{L^{2}(\mathbb{M}_{d},d\nu,\mathbb{R})},\dots,\left\langle \beta_{p},S_{n,k}^{d}\right\rangle_{L^{2}(\mathbb{M}_{d},d\nu,\mathbb{R})}\right)^{T}\nonumber\\
&&\boldsymbol{\varepsilon}_{n,k}=\left(\varepsilon_{n,k}(1),\dots,\varepsilon_{n,k}(T)\right)^{T}
\nonumber\\
&&=
\left(\left\langle \varepsilon_{1}, S_{n,k}^{d}\right\rangle_{L^{2}(\mathbb{M}_{d},d\nu ,\mathbb{R})},\dots,\left\langle \varepsilon_{T}, S_{n,k}^{d}\right\rangle_{L^{2}(\mathbb{M}_{d},d\nu ,\mathbb{R})}\right)^{T}.\nonumber\\
\label{femr}\end{eqnarray}
\subsection{GLS functional  parameter estimation}

According to equation (\ref{eq1}), applied to the case $\widetilde{X}_{t}=\varepsilon_{t},$ for every $t\in \mathbb{Z},$ one can consider the matrix sequence
\begin{eqnarray}&&
\left\{\boldsymbol{\Lambda}_{n}=\begin{bmatrix}
B_{n}(0) & \cdots & B_{n}(T-1)\\
\vdots & \vdots & \vdots\\
B_{n}(T-1) & \cdots & B_{n}(0)
\end{bmatrix},\quad n\in \mathbb{N}_{0}\right\}\nonumber\\
&&=\left\{\displaystyle{\int_{[-\pi,\pi]}}\scriptsize{\begin{bmatrix}
 f_{n}(\omega ) \quad & \cdots & \quad  \exp\left(i\omega (T-1)\right)f_{n}(\omega )\\
\quad \vdots  \quad &  \quad \vdots  \quad &\quad \vdots  \quad\\
\exp\left(i\omega (T-1)\right)f_{n}(\omega ) \quad & \cdots & \quad f_{n}(\omega )
\end{bmatrix}}d\omega,\ n\in \mathbb{N}_{0}\right\},\nonumber\\
\label{femssd}
 \end{eqnarray}
\noindent where, as in equation (\ref{klexpc2}), $B_{n}(t),$ $t=0,\dots,T-1,$ $n\in \mathbb{N}_{0},$ denote  the time--varying  diagonal coefficients in the series  expansion of the
functional entries of $\mathbf{R}_{\boldsymbol{\varepsilon}\boldsymbol\varepsilon}.$
 In the subsequent development we will assume that
  $X_{t,j}\in \mathbb{R},$ $t=1,\dots,T,$ $j=1,\dots,p,$ are such that \begin{equation}\sum_{n\in \mathbb{N}_{0}}\delta (n,d)\left(\mathbf{X}^{T}\boldsymbol{\Lambda}_{n}^{-1}\mathbf{X}\right)^{-1}<\infty.\label{fs}
   \end{equation}
   Note that under conditions in Theorem  4   in  \cite{MaMalyarenko} and  Lemma \ref{lemmat5ma}, from Cauchy--Schwartz inequality, $$\sum_{n\in\mathbb{N}_{0}}\delta(n,d)\boldsymbol{\Lambda}_{n}<\infty.$$

The  GLS functional parameter estimator of  $\boldsymbol{\beta}=\left[\beta_{1},\beta_{2},\dots,\beta_{p}\right]^{T}$ is computed from its projections into  the orthonormal basis
 $\{ S_{n,j}^{d},\ j=1,\dots,\delta(n,d),\ n\in \mathbb{N}_{0}\}$ of eigenfunctions of the Laplace Beltrami operator
$\Delta_{d}$ on $L^{2}(\mathbb{M}_{d},d\nu,\mathbb{R}).$
Specifically, from equation (\ref{ecseriesbeta}),  the GLS  $\widehat{\boldsymbol{\beta}}$  is the minimizer  of the  loss function
\begin{eqnarray}
L&=& \left\|\widetilde{\mathbf{Y}}-\mathbf{X}\boldsymbol{\beta }\right\|^{2}_{\mathbf{R}^{-1}_{\boldsymbol{\varepsilon}\boldsymbol{\varepsilon}}}
 =\sum_{n\in \mathbb{N}_{0}}\sum_{j=1}^{\delta (n,d)}\left[\widetilde{\mathbf{Y}}_{n,j}-\mathbf{X}\boldsymbol{\beta}_{n,j}\right]^{T}\boldsymbol{\Lambda}_{n}^{-1}
\left[\widetilde{\mathbf{Y}}_{n,j}-\mathbf{X}\boldsymbol{\beta}_{n,j}\right]\nonumber\\
&=&\sum_{n\in \mathbb{N}_{0}}\sum_{j=1}^{\delta (n,d)}\left\|\boldsymbol{\varepsilon}_{n,j}\right\|^{2}_{\boldsymbol{\Lambda}_{n}^{-1}},
\label{lfgls}
\end{eqnarray}
\noindent where, as before,  $\mathbf{X}=(X_{t,h})_{t=1,\dots,T; h=1,\dots,p},$
and    for $j=1,\dots,\delta(n,d)$  and  $n\in \mathbb{N}_{0},$   $\widetilde{\mathbf{Y}}_{n,j},$
$ \boldsymbol{\beta}_{n,j}$  and $\boldsymbol{\varepsilon}_{n,j}$ have been respectively  introduced in equations  (\ref{fclr}) and (\ref{femr}).
  In equation (\ref{lfgls}), for each $n\in \mathbb{N}_{0},$  $\boldsymbol{\Lambda}_{n}^{-1}$ denotes the matrix defining  the bilinear form  characterizing the inner product of the Reproducing Kernel Hilbert Space (RKHS) of $\left(\boldsymbol{\varepsilon}_{n,j},\ j=1,\dots,\delta(n,d)\right).$ Thus,

\begin{equation}
\label{ecbetamcg}
\widehat{\boldsymbol{\beta}}_{n,j}
=(\mathbf{X}^{T}\boldsymbol{\Lambda}^{-1}_{n}\mathbf{X})^{-1}\mathbf{X}^{T}\boldsymbol{\Lambda}_{n}^{-1}\widetilde{\mathbf{Y}}_{n,j},\quad j=1,\dots,\delta (n,d), \  n\in\mathbb{N}_{0}.
\end{equation}

Our predictor of the response is then given by:

\begin{eqnarray}
&&\widehat{\mathbf{Y}}(\mathbf{y})=\mathbf{H}\left(\mathbf{X}(\widehat{\boldsymbol{\beta }})\right)(\mathbf{y}),\quad \mathbf{y}\in \mathbb{M}_{d},
\label{eqprednl}
\end{eqnarray}
\noindent where
\begin{eqnarray}&&\widehat{\boldsymbol{\beta }}(\mathbf{y})=  \sum_{n\in \mathbb{N}_{0}}\sum_{j=1}^{\delta (n,d)}\widehat{\boldsymbol{\beta}}_{n,j}
 S_{n,j}^{d}(\mathbf{y})\nonumber\\
 &&=\left(\sum_{n\in \mathbb{N}_{0}}\sum_{j=1}^{\delta(n,d)}\widehat{\beta}_{n,j}^{(1)} S_{n,j}^{d}(\mathbf{y}),\dots,\sum_{n\in \mathbb{N}_{0}}\sum_{j=1}^{\delta(n,d)}\widehat{\beta}_{n,j}^{(p)} S_{n,j}^{d}(\mathbf{y})
 \right)^{T},\quad \mathbf{y}\in \mathbb{M}_{d}.
\label{eqprednlbb2}
\end{eqnarray}
\subsection{Moment properties of the GLS functional parameter estimator}
\label{secpropertbetaest}
The following proposition  provides the functional second--order moments of the  GLS parameter estimator $\widehat{\boldsymbol{\beta }}$ of $\boldsymbol{\beta }$ computed
in (\ref{eqprednlbb2}).
\begin{proposition}

 The following identities hold:
 \begin{itemize}
 \item[(i)]
 $E[\hat{\boldsymbol{\beta}}_{n,j}]=\boldsymbol{\beta}_{n,j},$ $j=1,\dots,\delta(n,d),$  $n\in \mathbb{N}_{0},$   i.e.,  $E\left[\widehat{\boldsymbol{\beta }}\right]=\boldsymbol{\beta }.$

\item[(ii)] $\mbox{Var}\left[\hat{\boldsymbol{\beta}}_{n,j}\right]=(\mathbf{X}^{T}\boldsymbol{\Lambda}_{n}^{-1}\mathbf{X})^{-1},$ $j=1,\dots,\delta(n,d),$  $n\in \mathbb{N}_{0},$  i.e.,\linebreak
$\mbox{Var}(\widehat{\boldsymbol{\beta }})=\sum_{n\in \mathbb{N}_{0}}\delta (n,d)\left(\mathbf{X}^{T}\boldsymbol{\Lambda}_{n}^{-1}\mathbf{X}\right)^{-1}.$
 \end{itemize}
\end{proposition}
\begin{proof}

The proof of (i) and (ii) follows straightforward as in the real--valued case. Specifically,

\begin{eqnarray}&&
E\left[ \hat{\boldsymbol{\beta}}_{n,j}\right]= E\left[(\mathbf{X}^{T}\boldsymbol{\Lambda}_{n}^{-1}\mathbf{X})^{-1}\mathbf{X}^{T}\boldsymbol{\Lambda}_{n}^{-1}
\widetilde{\mathbf{Y}}_{n,j}\right]\nonumber\\
&&=(\mathbf{X}^{T}\boldsymbol{\Lambda}_{n}^{-1}\mathbf{X})^{-1}\mathbf{X}^{T}\boldsymbol{\Lambda}_{n}^{-1}
E\left[\widetilde{\mathbf{Y}}_{n,j}\right]\nonumber\\
&&=(\mathbf{X}^{T}\boldsymbol{\Lambda}_{n}^{-1}\mathbf{X})^{-1}\mathbf{X}^{T}\boldsymbol{\Lambda}_{n}^{-1}\mathbf{X}\boldsymbol{\beta}_{n,j}\nonumber\\
&&=\boldsymbol{\beta}_{n,j},\quad j=1,\dots,\delta(n,d),\  n\in \mathbb{N}_{0}.
\label{ubbb}
\end{eqnarray}

Hence, from equation (\ref{ubbb}) and Fubini--Tonelli Theorem,  for every $\mathbf{y}\in \mathbb{M}_{d},$

\begin{eqnarray}
&& E\left[\hat{\boldsymbol{\beta}}(\mathbf{y})\right]=E\left[\left(\sum_{n=0}^{\infty}\sum_{k=1}^{\delta(n,d)}
\hat{\beta}_{n,k}^{(1)}
S^{d}_{n,k}(\mathbf{y}),\dots,\sum_{n=0}^{\infty}\sum_{k=1}^{\delta(n,d)}\hat{\beta}_{n,k}^{(p)}S^{d}_{n,k}(\mathbf{y})\right)^{T}\right]\nonumber\\
&&=\left(\sum_{n=0}^{\infty} \sum_{k=1}^{\delta(n,d)}E\left[\hat{\beta}_{n,k}^{(1)}\right] S^{d}_{n,k}(\mathbf{y}),\dots,\sum_{n=0}^{\infty} \sum_{k=1}^{\delta(n,d)}E\left[\hat{\beta}_{n,k}^{(p)}\right]S^{d}_{n,k}(\mathbf{y})\right)^{T}
\nonumber\\
&&=\left(\sum_{n=0}^{\infty}\sum_{k=1}^{\delta(n,d)}\beta_{n,k}^{(1)}S^{d}_{n,k}(\mathbf{y}),\dots,\sum_{n=0}^{\infty}\sum_{k=1}^{\delta(n,d)}
\beta_{n,k}^{(p)}S^{d}_{n,k}(\mathbf{y})\right)^{T}= \boldsymbol{\beta}(\mathbf{y}).\nonumber\\
\label{ub}
\end{eqnarray}

Regarding (ii),  as it is well known, since for every $j=1,\dots,\delta (n,d),$ $n\in \mathbb{N}_{0},$  \begin{equation}\widehat{\boldsymbol{\beta}}_{n,j}=\boldsymbol{\beta }_{n,j}+(\mathbf{X}^{T}\boldsymbol{\Lambda}_{n}^{-1}\mathbf{X})^{-1}
\mathbf{X}^{T}\boldsymbol{\Lambda}_{n}^{-1}\boldsymbol{\varepsilon}_{n,j},
\label{equb}
\end{equation}
\noindent we have

\begin{eqnarray}
\mbox{Var}\left[\widehat{\boldsymbol{\beta}}_{n,j}\right]&=&
E\left[ \left(\widehat{\boldsymbol{\beta}}_{n,j}-\boldsymbol{\beta}_{n,j}\right)^{T}\left(\hat{\boldsymbol{\beta}}_{n,j}-\boldsymbol{\beta}_{n,j}\right)\right]
\nonumber\\
&=&
(\mathbf{X}^{T}\boldsymbol{\Lambda}_{n}^{-1}\mathbf{X})^{-1}\mathbf{X}^{T}\boldsymbol{\Lambda}_{n}^{-1}\boldsymbol{\Lambda}_{n}
\boldsymbol{\Lambda}_{n}^{-1}\mathbf{X}(\mathbf{X}^{T}\boldsymbol{\Lambda}_{n}^{-1}\mathbf{X})^{-1}
\nonumber\\
&=&(\mathbf{X}^{T}\boldsymbol{\Lambda}_{n}^{-1}\mathbf{X})^{-1},\quad \forall j\in \{1,\dots,\delta(n,d)\},\quad n\in \mathbb{N}_{0}.
\label{vgls}
\end{eqnarray}

From (\ref{vgls}), applying uncorrelation of  the sequence of centered    random processes
$\left\{\varepsilon_{n,j}(t),\ t\in [0,T]\right\}$  (see equation \ref{equb})
\begin{eqnarray}
&&E\left[\left\|\widehat{\boldsymbol{\beta }}-\boldsymbol{\beta }\right\|_{L^{2}(\mathbb{M}_{d},d\nu,\mathbb{R})}^{2}\right]=
\sum_{n\in \mathbb{N}_{0}}\delta (n,d)(\mathbf{X}^{T}\boldsymbol{\Lambda}_{n}^{-1}\mathbf{X})^{-1}<\infty,\nonumber
\end{eqnarray}
\noindent  under condition (\ref{fs}).
\end{proof}
\subsection{Functional spectral based  plug--in estimation of $\boldsymbol{\beta}$}
\label{secmiss}
This section presents a plug--in GLS estimation methodology when the second order structure of the error term is unknown. In this case,  the entries of the matrix sequence $\left\{\boldsymbol{\Lambda}_{n},\ n\in \mathbb{N}_{0}\right\}$ are estimated
in the spectral domain  under the following semiparametric  modelling  (see  \cite{RuizMedina2022}):

\medskip

\noindent \textbf{Assumption I}. Assume that the entries $f_{n},$ $n\in \mathbb{N}_{0},$ of the matrix sequence in (\ref{femssd}) admit the following
semiparametric modeling, for every $n\in \mathbb{N}_{0},$
\begin{eqnarray}f_{n,\theta}(\omega)&=&
B_{n}^{\eta}(0)M_{n}(\omega )\left[4(\sin(\omega /2))^{2}\right]^{-\alpha (n,\theta)/2},\ \theta \in \Theta,\ \omega \in [-\pi,\pi],
\label{eqsmc1}\end{eqnarray}
\noindent where $\left\{\alpha(n,\theta),\ n\in \mathbb{N}_{0}\right\}$ are the diagonal coefficients involved in the series expansion of  kernel
$\mathcal{K}_{\mathcal{A}_{\theta }}$ of LRD operator  $\mathcal{A}_{\theta },$ which is assumed to hold in the following sense:
\begin{eqnarray}&&
\int_{\mathbb{M}_{d}}\mathcal{K}_{\mathcal{A}_{\theta }}(\mathbf{y},\mathbf{z})f(\mathbf{y})g(\mathbf{z})d\mathbf{y}d\mathbf{z}
\nonumber\\
&&=\int_{\mathbb{M}_{d}}\sum_{n\in \mathbb{N}_{0}} \alpha (n,\theta)
\sum_{j=1}^{\delta (n,d)}S_{n,j}^{d}(\mathbf{y})S_{n,j}^{d}(\mathbf{z})f(\mathbf{y})g(\mathbf{z})d\mathbf{y}d\mathbf{z},\nonumber
\end{eqnarray}
\noindent for every $f,g\in \mathcal{C}^{\infty}(\mathbb{M}_{d}),$ with  $\mathcal{C}^{\infty}(\mathbb{M}_{d})$ denoting the space of infinitely differentiable functions with compact support contained in  $\mathbb{M}_{d}.$
 Here,  $l_{\alpha }\leq \alpha (n,\theta )\leq L_{\alpha },$  for every $n\in \mathbb{N}_{0},$  and $\theta \in \Theta,$ with $l_{\alpha }, L_{\alpha }\in (0,1/2).$
 The elements of the sequence
 $\left\{B_{n}^{\eta}(0),\ n\in \mathbb{N}_{0}\right\}$ are the eigenvalues of the trace autocovariance operator $R^{\eta}_{0}$ of the innovation process  $\eta ,$ involved in the definition of the  error term. The operator family $\left\{\mathcal{M}_{\omega},\ \omega \in [-\pi,\pi]\right\}$ is included in the space of trace operators, and, hence, its elements admit the following series representation:
\begin{eqnarray}
&&\mathcal{K}_{\mathcal{M}_{\omega}}(\mathbf{y},\mathbf{z})=\sum_{n\in \mathbb{N}_{0}} M_{n}(\omega )
\sum_{j=1}^{\delta (n,d)}S_{n,j}^{d}(\mathbf{y})S_{n,j}^{d}(\mathbf{z}),\quad \mathbf{y},\mathbf{z}\in \mathbb{M}_{d},\nonumber
\end{eqnarray}
\noindent in the norm of the space $\mathcal{S}(L^{2}(\mathbb{M}_{d}, d\nu, \mathbb{C}))$ of Hilbert--Schmidt operators on $L^{2}(\mathbb{M}_{d}, d\nu, \mathbb{C}).$ Note that in the particular case where  $\alpha (n,\theta)=0,$ for every $n\in \mathbb{N}_{0},$ and $\theta \in \Theta ,$ $X$ displays SRD under the condition \begin{equation}
\sum_{\tau\in \mathbb{Z}}\sum_{n\in \mathbb{N}_{0}}\delta (n,d)\left|\int_{-\pi}^{\pi}\exp(i\omega \tau)M_{n}(\omega )d\omega \right|<\infty.
\label{ftracetime}
\end{equation}

  \medskip

  We apply the minimum contrast estimation strategy introduced in equations (5.1)--(5.19) in \cite{RuizMedina2022}, and
  equations (3.8)--(3.16) in \cite{Ovalle23}, considering the special case of $H=L^{2}(\mathbb{M}_{d},d\nu,\mathbb{C}).$
  Specifically, parameter $\theta$ in equation (\ref{eqsmc1}), characterizing the pure point spectrum of LRD operator $\mathcal{A}_{\theta },$
  is estimated by $\widehat{\theta }_{T}$ satisfying
  \begin{eqnarray}
\widehat{\theta }_{T}&=& \mbox{arg} \ \min_{\theta \in \Theta }
\left\|-\int_{-\pi}^{\pi }p_{\omega }^{(T)}
\ln\left(\Upsilon_{\omega ,\theta }\right)\mathcal{W}_{\omega} d\omega\right\|_{\mathcal{L}(L^{2}(\mathbb{M}_{d},d\nu; \mathbb{C}))}, \label{mfetheta}
\end{eqnarray}
  \noindent where  for each $\theta \in \Theta,$  and $\omega \in [-\pi,\pi],$  $\omega \neq 0,$
  \begin{equation}
 \Upsilon_{\omega ,\theta }=[\mathcal{N}_{\theta }]^{-1}\mathcal{F}_{\omega ,\theta }= \mathcal{F}_{\omega ,\theta }[\mathcal{N}_{\theta }]^{-1},
\label{fsdo}
\end{equation}
\noindent with, as before, $\left\{\mathcal{F}_{\omega,\theta },\ \omega \in [-\pi,\pi]\right\}$  being the  spectral density operator  family.
Operator $\mathcal{N}_{\theta }$ has kernel
  \begin{eqnarray}&&
\mathcal{K}_{\mathcal{N}_{\theta }}(\mathbf{y},\mathbf{z})
=\sum_{n\in \mathbb{N}_{0}}\widetilde{W}(n)\left[\int_{-\pi}^{\pi} \frac{B_{n}^{\eta }(0)M_{n}(\omega )\left[4(\sin(\omega /2))^{2}\right]^{-\alpha (n,\theta)/2}
}{|\omega|^{-\gamma }} d\omega \right]\nonumber\\
&&\hspace*{0.5cm}\times
\sum_{j=1}^{\delta (n,d)}S_{n,j}^{d}(\mathbf{y})S_{n,j}^{d}(\mathbf{z}),\quad \mathbf{y},\mathbf{z}\in \mathbb{M}_{d},\ \theta \in \Theta,
\label{eqdeftildew2}
\end{eqnarray}
\noindent where $\widetilde{W}$ denotes the  positive self--adjoint operator on $L^{2}(\mathbb{M}_{d}, d\nu,\mathbb{C})$ factorizing the weighting operator $\mathcal{W}_{\omega }=\widetilde{W}|\omega |^{\gamma },$  for every $\omega \in [-\pi,\pi],$ and $\gamma >0.$
 Fourier transform inversion formula leads to the corresponding  estimation
$$\widehat{B}_{n,\hat{\theta}_{T}}(t)=\int_{-\pi}^{\pi}\exp(i\omega t)f_{n,\widehat{\theta}_{T}}(\omega)d\omega,\quad n\in \mathbb{N}_{0},$$
\noindent  of the entries of $\boldsymbol{\Lambda}_{n,\widehat{\theta}_{T}},$ given by

\[
\boldsymbol{\Lambda}_{n,\widehat{\theta}_{T}}=\left(\begin{bmatrix}
\widehat{B}_{n,\widehat{\theta}_{T}}(0) & \cdots & \widehat{B}_{n,\widehat{\theta}_{T}}(T-1)\\
\vdots & \vdots & \vdots\\
\widehat{B}_{n,\widehat{\theta}_{T}}(T-1) & \cdots & \widehat{B}_{n,\widehat{\theta}_{T}}(0)
\end{bmatrix}\right),\quad n\in \mathbb{N}_{0}.
\]
\noindent Thus, for every $n\in \mathbb{N}_{0},$
\begin{eqnarray}&&
\widehat{\boldsymbol{\beta}}_{n, j, \widehat{\theta}_{T}}=\left(\mathbf{X}^{T}
\boldsymbol{\Lambda}^{-1}_{n,\widehat{\theta}_{T}}\mathbf{X}\right)^{-1}\mathbf{X}^{T}
\boldsymbol{\Lambda}^{-1}_{n,\widehat{\theta}_{T}}\widetilde{\mathbf{Y}}_{n,j},\quad n\in \mathbb{N}_{0},\ j=1,\dots, \delta(n,d),\nonumber
\label{ent}
\end{eqnarray}
\noindent
 and the corresponding  plug--in nonlinear predictor is  computed  as
\begin{equation}
\widehat{\mathbf{Y}}_{\widehat{\theta}_{T}}(\mathbf{y})=
\mathbf{H}\left(\mathbf{X}(\widehat{\boldsymbol{\beta}}_{\widehat{\theta}_{T}})\right)(\mathbf{y}),\quad \mathbf{y}\in\mathbb{M}_{d},
\label{emppred}
\end{equation}
\noindent where
\begin{eqnarray}&&\widehat{\boldsymbol{\beta }}_{\widehat{\theta}_{T}}(\mathbf{y})=  \sum_{n\in \mathbb{N}_{0}}\sum_{j=1}^{\delta (n,d)}\widehat{\boldsymbol{\beta}}_{n,j,\widehat{\theta}_{T}}
 S_{n,j}^{d}(\mathbf{y})\nonumber\\
 &&=\left(\sum_{n\in \mathbb{N}_{0}}\sum_{j=1}^{\delta(n,d)}\widehat{\beta}_{n,j,\widehat{\theta}_{T}}^{(1)} S_{n,j}^{d}(\mathbf{y}),\dots,\sum_{n\in \mathbb{N}_{0}}\sum_{j=1}^{\delta(n,d)}\widehat{\beta}_{n,j,\widehat{\theta}_{T}}^{(p)} S_{n,j}^{d}(\mathbf{y})
 \right)^{T},\quad \mathbf{y}\in \mathbb{M}_{d}.
\label{eqprednlbb}
\end{eqnarray}

\section{Simulations}

\label{secsimulation}
The performance of the proposed nonlinear  multiple functional regression predictor is illustrated for $\mathbb{M}_{d}=\mathbb{S}_{d}=\{ \mathbf{y}\in \mathbb{R}^{d+1}; \ \|\mathbf{y}\|=1\},$  under a log--Gaussian scenario, and  for $\mathbf{H}(\mathbf{X}(\boldsymbol{\beta })(\mathbf{y}))=\sum_{k=0}^{\infty}\frac{(\mathbf{X}(\boldsymbol{\beta }(\mathbf{y})))^{k}}{k!},$ with
$$(\mathbf{X}(\boldsymbol{\beta }(\mathbf{y})))^{k}=\left(\left(\sum_{j=1}^{p}X_{1,j}\beta_{j}(\mathbf{y})\right)^{k},\dots, \left(\sum_{j=1}^{p}X_{T,j}\beta_{j}(\mathbf{y})\right)^{k}\right)^{T}.$$
  In the  Supplementary Material, the linear and Gaussian case is considered.   In Sections \ref{secsimulationch5}  and \ref{secMISSsim}, the theoretical and plug--in GLS predictors are  respectively computed by projection into  the  spherical harmonics basis.

\subsection{Theoretical predictor}
\label{secsimulationch5}

The regression prediction results are tested for functional sample sizes  $T=110, 300, 500.$   We consider the case where the covariates are strong correlated. Specifically, their dynamic at each one of the two spherical scales selected is represented in terms of fractional Brownian motion with respective Hurst parameter values  $H=0.5/k,$  $k=1,2.$  In the generations, we have considered   Matlab function \emph{wfbm}, based on the wavelet transform  (see, e.g., \cite{AbrySellan1996}). Our choice of the spherical functional regression parameters is given by the    eigenfunctions $S^{2}_{1,1},$ and  $S^{2}_{1,2}$ of the spherical  Laplace Beltrami operator,   displayed at the two plots of the first line of Figure  18 in Section 2.1 of the Supplementary Material. The regression error is generated from its truncated   expansion  (see Figure  19  in Section 2.1 of the Supplementary Material, where realization 75 is showed). This realization corresponds to the projected process  into the subspace generated by the eight eigenfunctions plotted in  Figure  18   of Section  2.1 of the Supplementary Material. The corresponding  time varying coefficients are computed from  the inverse Fourier transform of the  square root of the frequency--varying eigenvalues (\ref{eqsmc1}) of the  elements of the spectral density operator family  under \textbf{Assumption I} (see left--hand--side of Figure \ref{Fig1ts}).

\begin{figure}[H]
\begin{center}
\includegraphics[height=0.3\textheight, width=0.45\textwidth]{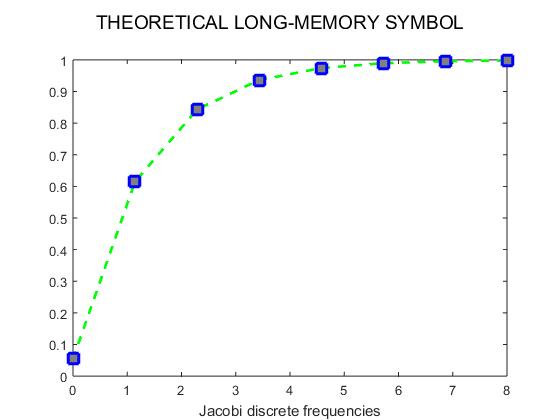}
\includegraphics[height=0.3\textheight, width=0.45\textwidth]{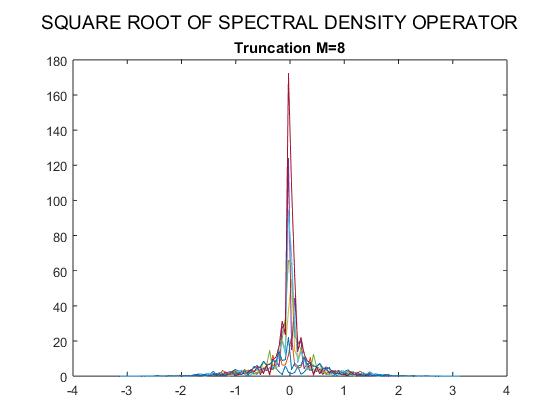}
\caption{LRD operator eigenvalues  (left--hand side), and the squared root of  frequency varying eigenvalues in 50th realization, for the first 8  Laplace Beltrami operator eigenspaces selected (right--hand side)}
\label{Fig1ts}
\end{center}
\end{figure}
The   nonlinear transformation  of the computed  truncated version of the GLS functional  linear predictor   approximates equation (\ref{eqprednl}). See Figures \ref{Figure5ts} and \ref{Figure6ts} where one realization of the nonlinear response and its functional regression prediction are respectively showed.
The corresponding empirical mean absolute errors are obtained from  $R=100$ repetitions of functional samples of
sizes $T=110,300,500.$     The results are displayed  in Figures  20 and  21 in Section  2.1 of the Supplementary Material for the functional sample sizes $T=110,300,$   and for   $T=500$ in  Figure \ref{Figure7tst500}.  One can observe  an important reduction of such  errors  as  the functional sample size  increases.

\begin{figure}[H]
\begin{center}
\includegraphics[height=0.9\textheight, width=\textwidth]{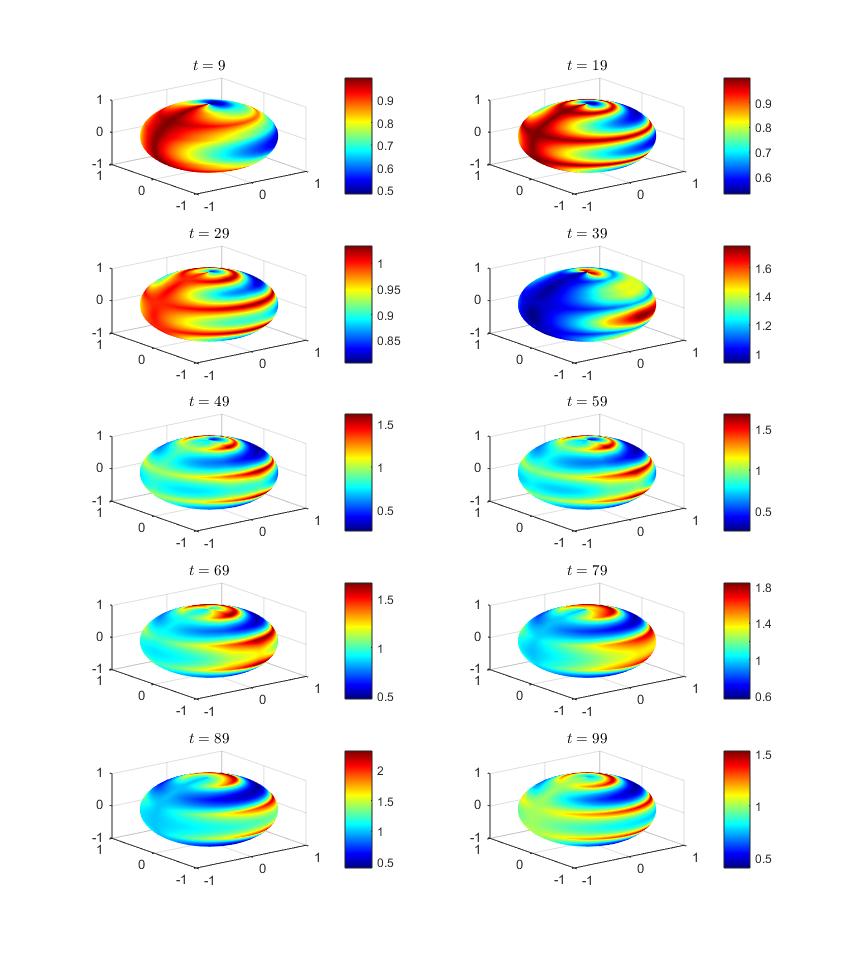}
\caption{Log--Gaussian nonlinear spherical functional response values at times $t=9,19,29,39,49,59,69,79,89,99$ (corresponding to realization 75)}
\label{Figure5ts}
\end{center}
\end{figure}

\begin{figure}[H]
\begin{center}
\includegraphics[height=0.9\textheight, width=\textwidth]{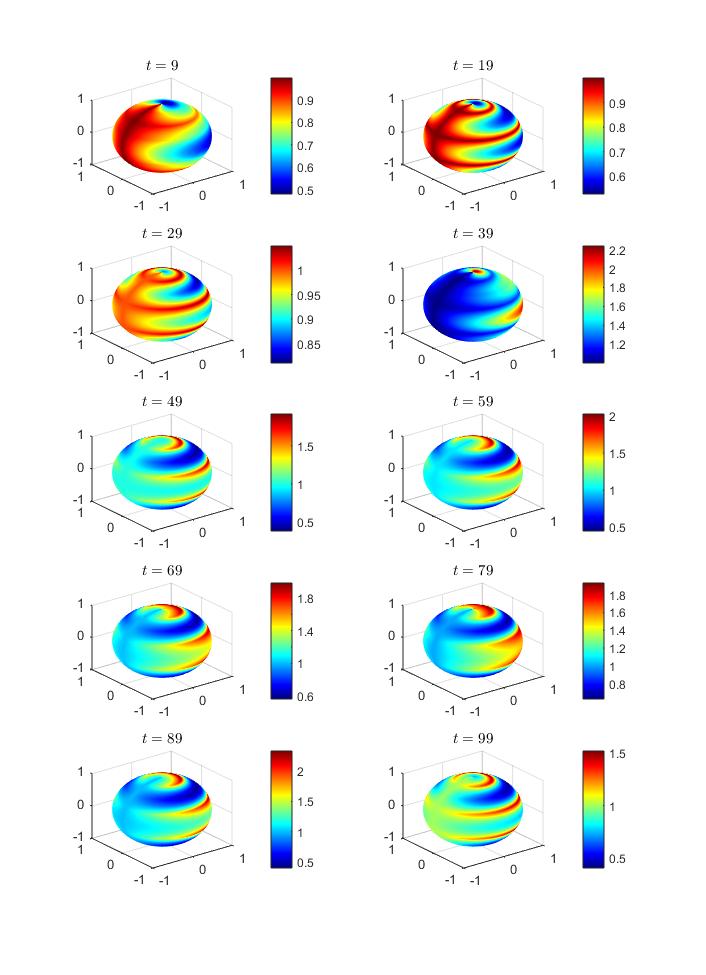}
\caption{Nonlinear  spherical functional response   predictions at times $t=9,19,29,39,49,59,69,79,89,99$ (corresponding to realization 75)}
\label{Figure6ts}
\end{center}
\end{figure}

\begin{figure}[H]
\begin{center}
\includegraphics[height=0.9\textheight, width=\textwidth]{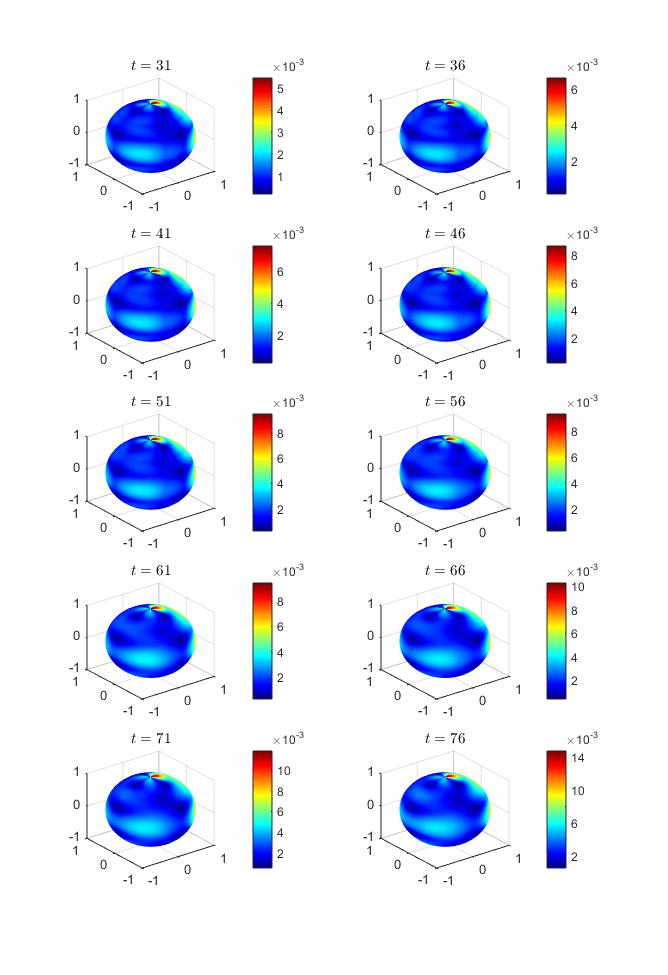}
\caption{Empirical mean absolute errors, based on 100 repetitions, for functional sample size $T=500,$  associated with the spherical functional regression predictor at times $t=31,36,41,46,51,56, 61, 66, 71, 76$}
\label{Figure7tst500}
\end{center}
\end{figure}


\subsection{Plug--in predictor}
\label{secMISSsim}
This section analyzes the case where the covariance operator family of the error term is unknown.  As before, the residual analysis is performed in the double spectral domain, implementing the minimum contrast estimation of the second--order structure of the error term displaying spatial scale varying LRD. Two  cases are analyzed respectively corresponding to   increasing LRD eigenvalue sequence  (see left--hand side of Figure \ref{Fig1ts}),  and  decreasing  LRD eigenvalue sequence (see Section 2.2.1 of the Supplementary Material).

In this section,  in the generation of the temporal covariates, we consider  the Hurst parameter value $H=0.001.$
This parameter value is close to the lowest bound of the interval $(0,1/2),$
where negative long--term correlation is displayed by fractional Brownian motion. Thus,   the generated  temporal dynamics is  very  far from the independent increment dynamics  of   Brownian motion.
      This feature makes the results in our simulation study more attractive, since reveal a good performance of our approach in the most complex case, where plug--in prediction is implemented under temporal  strong--dependence  of the  covariates and of  the  functional error term. Our choice here of the spherical functional regression parameters corresponds to the eigenfunctions $S^{2}_{1,1},$ and  $S^{2}_{2,1}$ plotted at the left--hand side of the first two lines of Figure  18 in Section 2.1 of the Supplementary Material.

 Minimum contrast estimation in the functional spectral domain  is implemented to approximate the second--order  structure of the LRD isotropic spherical functional  error term in the spectral domain  (see \cite{Ovalle23}; \cite{RuizMedina2022}).  In this implementation, we consider a set of $100$ candidates for the first eight eigenvalues of the LRD operator (see Figure 22 in Section 2.2 in the Supplementary Material).

The $50$th realization of  the generated spherical functional error term, and its spectral based minimum contrast estimation are plotted in Figures 23  and 24  in Section  2.2 of the Supplementary Material, respectively. The  empirical mean absolute errors, based on $100$ repetitions  (see Figure  25 in   Section 2.2 of the Supplementary Material),  and based on $500$ repetitions (see Figure \ref{f3IncremissR500}), associated with the minimum contrast estimates of the error term, are computed.  When the number of repetitions increases from $100$ to $500,$ a substantial improvement is observed, reflected in  the reduced extension of the spherical  areas with the highest values of the empirical mean absolute errors.
\begin{figure}[H]
\begin{center}
\includegraphics[height=0.9\textheight, width=\textwidth]{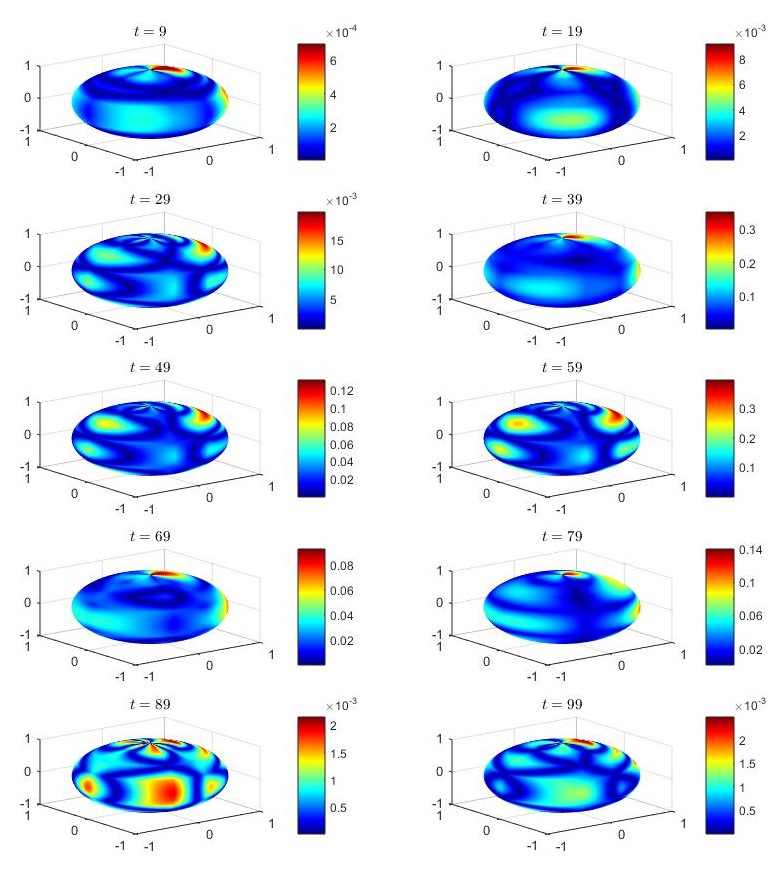}
\caption{Empirical mean absolute errors, based on $500$ repetitions, associated with the minimum contrast estimator of the spherical functional error term  at times $t=9,19,29,39,49,59,69,79,89,99$}
\label{f3IncremissR500}
\end{center}
\end{figure}

\begin{figure}[H]
\begin{center}
\includegraphics[height=0.9\textheight, width=\textwidth]{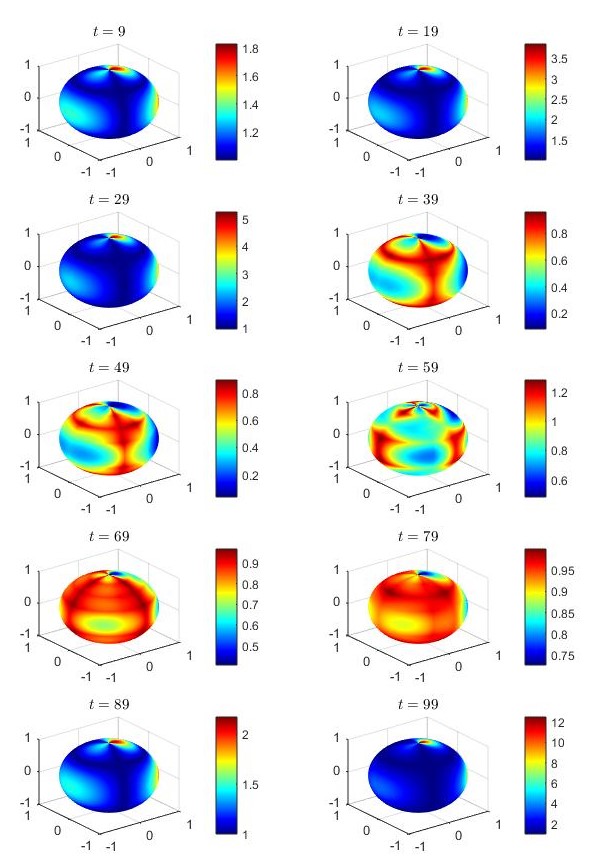}
\caption{Spherical functional response values at times $t=9,19,29,39,49,59,69,79,89,99$ (corresponding to realization 50)}
\label{fNLRESPINCRMISS}
\end{center}
\end{figure}
\begin{figure}[H]
\begin{center}
\includegraphics[height=0.9\textheight, width=\textwidth]{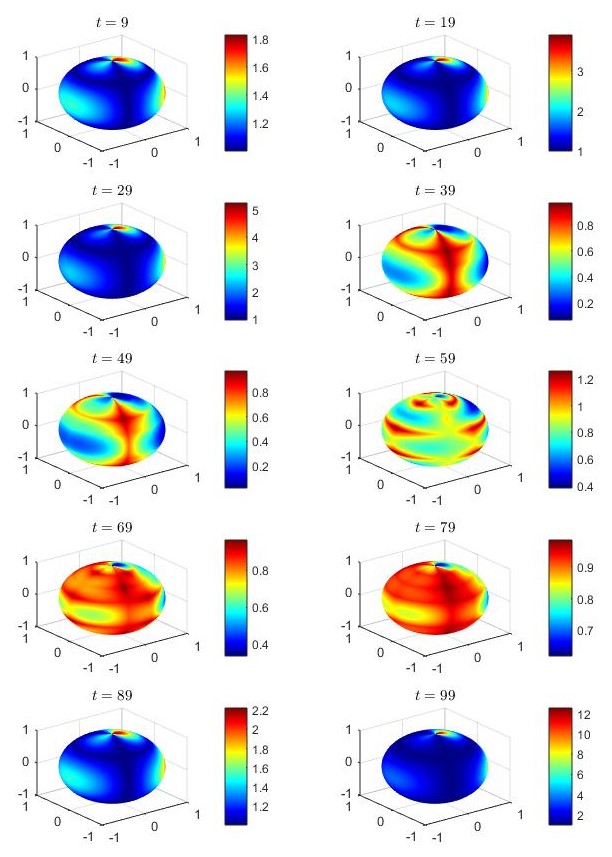}
\caption{Spherical functional response nonlinear  regression  predictions at times $t=9,19,29,39,49,59,69,79,89,99$ (corresponding to realization 50)}
\label{fNLPLUGINCRE}
\end{center}
\end{figure}

\begin{figure}[H]
\begin{center}
\includegraphics[height=0.88\textheight, width=\textwidth]{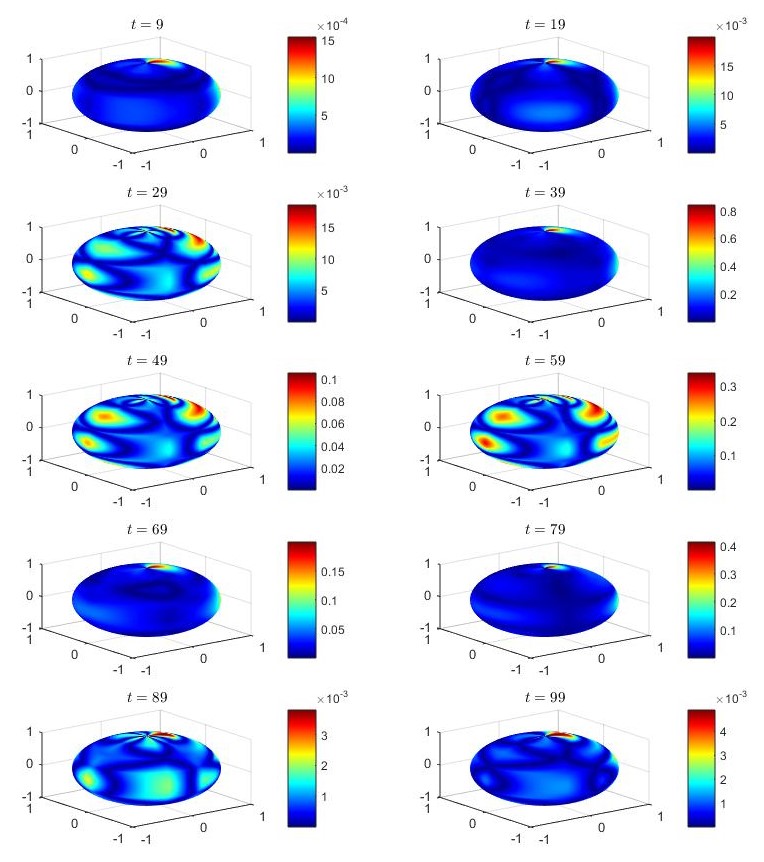}
\caption{Empirical mean absolute errors, based on 500 repetitions,  associated with the spherical functional nonlinear  regression predictor at times $t=9,19,29,39,49,59,69,79,89,99$}
\label{f4IncremissR500}
\end{center}
\end{figure}

The plug--in  spherical functional regression predictor (\ref{emppred}) is also computed.  Under this misspecified scenario a larger number of repetitions and functional sample sizes are required to improve the performance of the estimates computed from a truncated version of (\ref{emppred}). That is the reason why to illustrate the finite--sample--size  performance of  our plug-in functional regression predictor we have displayed the results for the functional sample size $T=110,$ considering $R=100$ (see Figure  26   in Section 2.2 of the Supplementary Material), and $R=500$ (see Figure \ref{f4IncremissR500}) repetitions. See also Figures \ref{fNLRESPINCRMISS} and \ref{fNLPLUGINCRE}, and  Section 2.2.1 of the Supplementary Material, where similar results are plotted for the decreasing LRD  operator eigenvalue sequence.


\section{Real--data application}
\label{secApplic}
This section considers the implementation of the proposed nonlinear spherical functional regression methodology in the prediction of the time evolution of downward solar radiation flux  earth maps, from the daily observation on the earth globe of atmospheric pressure  at high cloud bottom. A synthetic data set is generated   based on the nonlinear physical equations  governing the coupled dynamics of  both physical  magnitudes.

The nonlinear spatiotemporal mean of  the generated   downward solar radiation flux during the period  autumn--winter  is displayed in Figure \ref{f1app}. We summarize the main steps followed in the computation of this magnitude. A
 starting polar and azimuthal  angle grid  with  $180$ nodes   in the intervals $(0,\pi ),$   and  $(0,2\pi )$ is considered, and its associated  meshgrid in the corresponding two--dimensional angle interval is also implemented. The polar angle values are converted into latitudes  for the computation of the  Zenith Angle (ZA), which is one of the  input variables  of the physical equation defining Solar Irradiance (SI). Note that the  ZA  depends on the time of the year, and on the declination through a suitable trigonometric equation. The declination is given by a sinusoidal function also depending  on the day of the year.  Other parameters involved in these previous  physical equations  are  the Earth Radius  $ER=6371000$ in meters, and the Solar Constant  $G_{0} = 1361$ in W/$m^{2}$.

The SI is obtained  from the Clear Sky Index (CSI=0.8) by using the relationship

\begin{equation}
SI = G_{0}(CSI) \cos(ZA)/ \pi.
\label{ephsi}
\end{equation}

Finally, to reflect persistent in time of SI random fluctuations during autumn--winter, an  LRD isotropic spherical functional process is generated as error term, with time--varying independent coefficients defined from fractional Brownian motion, suitable scaled with  the solar irradiance standard deviation value 160.2262 (see Figure \ref{f2app}).

The nonlinear spatiotemporal mean  of the  atmospheric pressure  is computed (see Figure 32 of the Supplementary Material for spring--summer period, and  Figure \ref{f1appd2} for autumn--winter period),   from the  barometric equation, involving sea level pressure $P_{0}=1013.25,$ air molar mass  $M=0.029$ in  kilograms per mole, acceleration due gravity $g=9.81$ in m/s$^{2},$ ideal gas constant $RC=8.314,$ Kelvin temperature $TT=273+15,$ and usual range of heights at  bottom of high cloud, where  we have considered the height  interval (6000, 12000) in meters. Thus, pressures $pp$ obey the equation  $$pp=P_{0}(\exp(-M(g)(heights)/(RC(TT)))).$$  Again, a meshgrid is constructed from latitude and days to finally compute the daily  values of the    spherical functional isotropic regressor mean over a year from the input argument $pp,$ in terms of polar angle, amplitude of pressure variation with latitude and over days, and angular frequency corresponding to an annual cycle.  We have considered the value $49.6453$ of pressure standard deviation in the scaling of the  LRD isotropic spherical functional  time series generated to model temporal persistence of random fluctuations  (see Figure \ref{f3app}). The final generations are obtained by adding to the spatiotemporal mean computed, the  generated   spatiotemporal  LRD isotropic spherical process  (see Figure \ref{f4app}).

Note that although this synthetic spherical functional data set has been generated for the time period of one year, for illustration purposes, we have restricted our attention to the period autumn--winter, where low pressure is frequently observed at earth globe areas of medium and   high latitudes in both hemispheres, while the highest pressures are localized at tropical and subtropical areas.  The reverse situation  corresponds to the spring--summer period
 (see Figure 32 in Section 3 of the Supplementary Material). Indeed, this fact constitutes one of our main motivations to include in this nonlinear spherical functional regression problem the temporal information. Specially, regarding time--varying  covariates in this example, one can see how  spherical patterns displayed  by the spherical  functional regressor change drastically in these two periods (autumn--winter and spring--summer), affecting in a very different way  the response defined by solar irradiance.

The results after  implementation of the proposed nonlinear spherical   multiple   functional regression predictor are showed  in Figure \ref{f6app} where the original values of the response are plotted at the left--hand side for different times, while at the right--hand side the corresponding spherical functional regression predictor values are showed. Note that, the spherical functional regression predictor reproduces the magnitudes and the spherical patterns  of the spherical functional solar irradiance very close.

\begin{figure}[H]
\begin{center}
\includegraphics[height=0.85\textheight, width=0.8\textwidth]{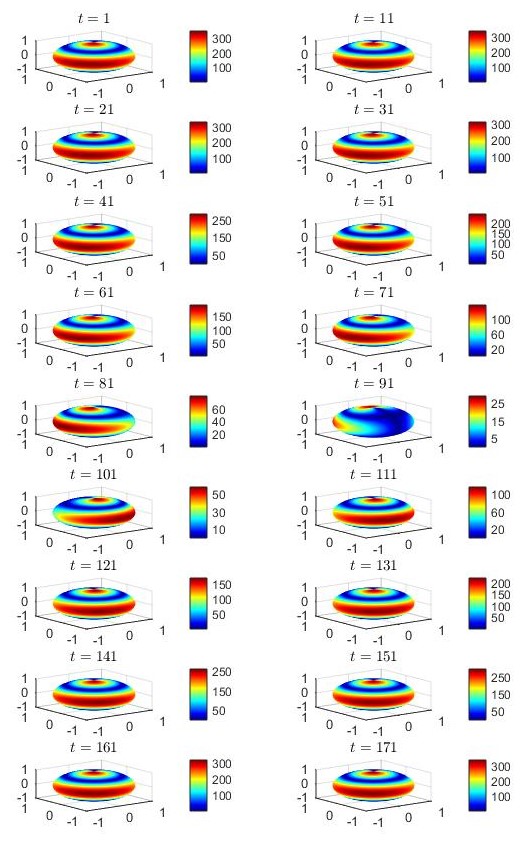}
\caption{Nonlinear response mean computed from evaluation of physical model (\ref{ephsi}) of downward solar radiation flux during autumn-winter.  Its spherical functional values are displayed at times \linebreak $t=1,11,21,31,41,51,61,71,81,91,101,111,121,131,141,151,161,171$}
\label{f1app}
\end{center}
\end{figure}

\begin{figure}[H]
\begin{center}
\includegraphics[height=0.95\textheight, width=0.8\textwidth]{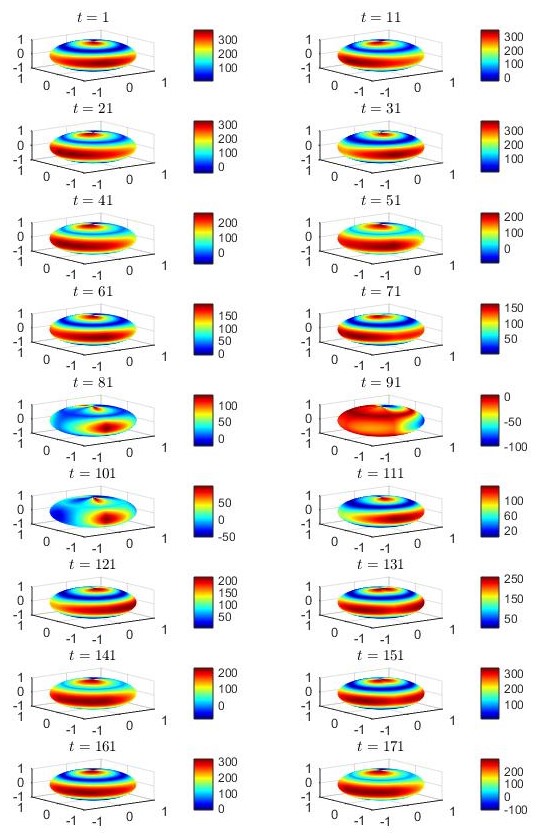}
\caption{Daily spherical functional response observations during autumn-winter period. Generated synthetic data of downward solar radiation flux are displayed at times \linebreak $t=1,11,21,31,41,51,61,71,81,91,101,111,121,131,141,151,161,171$}
\label{f2app}
\end{center}
\end{figure}

\begin{figure}[H]
\begin{center}
\includegraphics[height=0.95\textheight, width=0.8\textwidth]{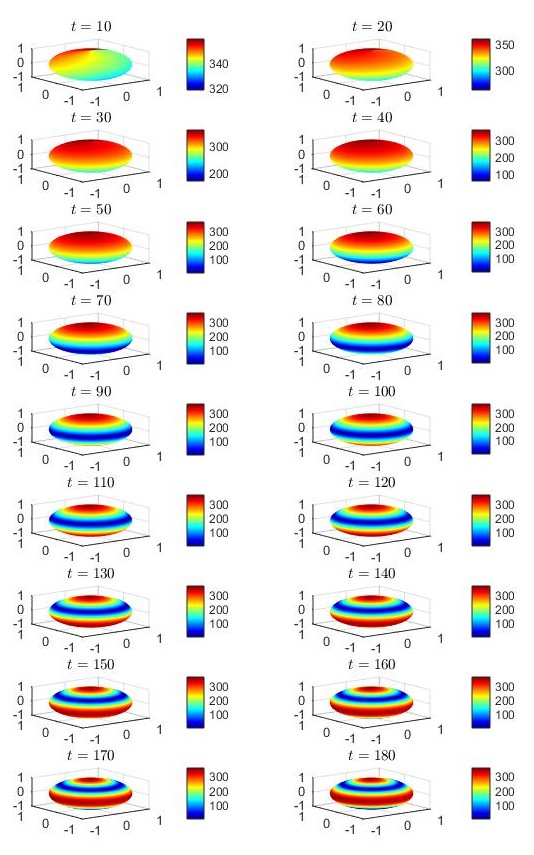}
\caption{Nonlinear spherical functional regressor  mean computed from barometric equation  during autumn-winter.  Its spherical functional values are displayed at times \linebreak $t=10, 20, 30, 40, 50, 60, 70, 80, 90, 100, 110,120,130,140,150, 160,170,180$}
\label{f1appd2}
\end{center}
\end{figure}

\begin{figure}[H]
\begin{center}
\includegraphics[height=0.95\textheight, width=0.8\textwidth]{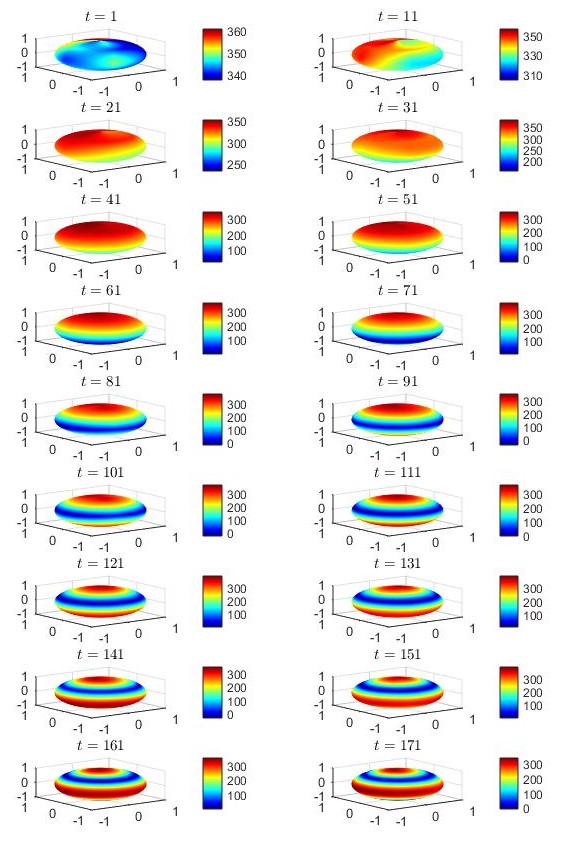}
\caption{Daily spherical functional regressor observations during autumn-winter period. Generated synthetic data   of atmospheric pressure at high
cloud bottom are displayed at times \linebreak $t=1,11,21,31,41,51,61,71,81,91,101,111,121,131,141,151,161,171$}
\label{f3app}
\end{center}
\end{figure}

\begin{figure}[H]
\begin{center}
\includegraphics[height=0.95\textheight, width=0.8\textwidth]{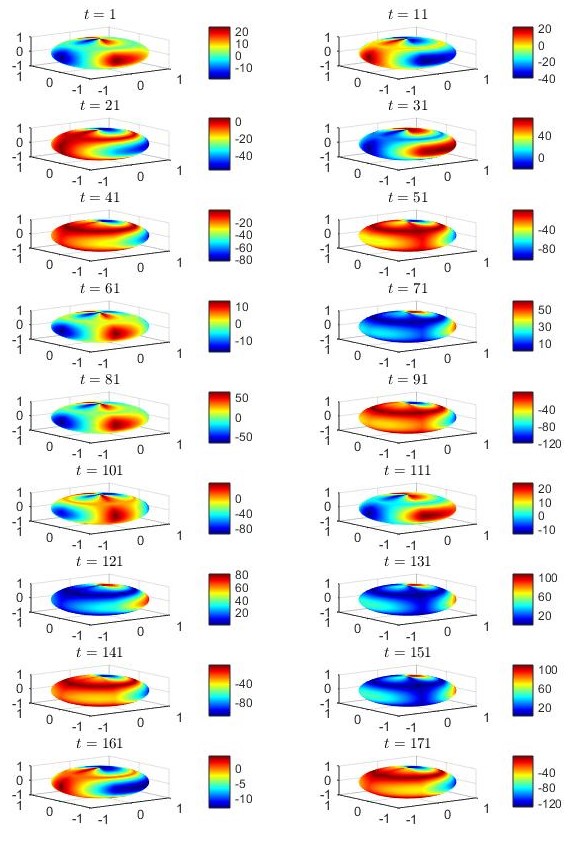}
\caption{Spatiotemporal LRD isotropic  spherical process at times   $t=1,11,21,31,41,51,61,71,$\linebreak $81,91,101,111,121,131,141,151,161,171$}
\label{f4app}
\end{center}
\end{figure}

\begin{figure}[H]
\begin{center}
\includegraphics[height=0.8\textheight, width=\textwidth]{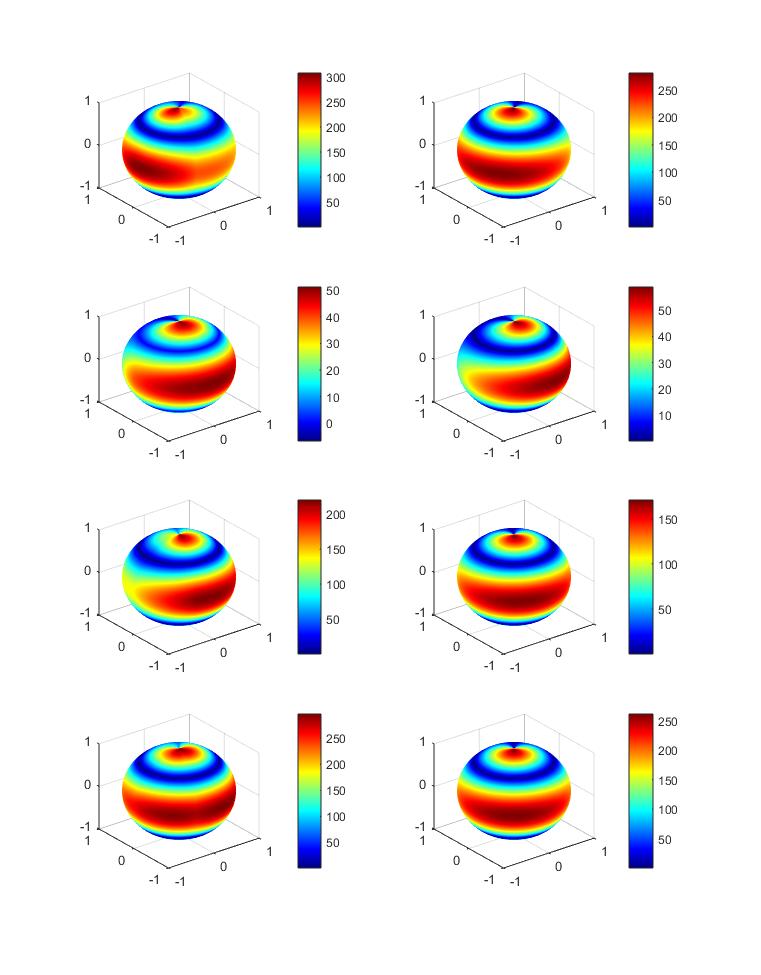}
\caption{Spherical functional response values (left-hand side) and spherical functional regression prediction values (right-hand side) at times $t=41,101,121, 141$ from top to the bottom respectively}
\label{f6app}
\end{center}
\end{figure}

\begin{figure}[H]
\begin{center}
\includegraphics[height=0.95\textheight, width=0.8\textwidth]{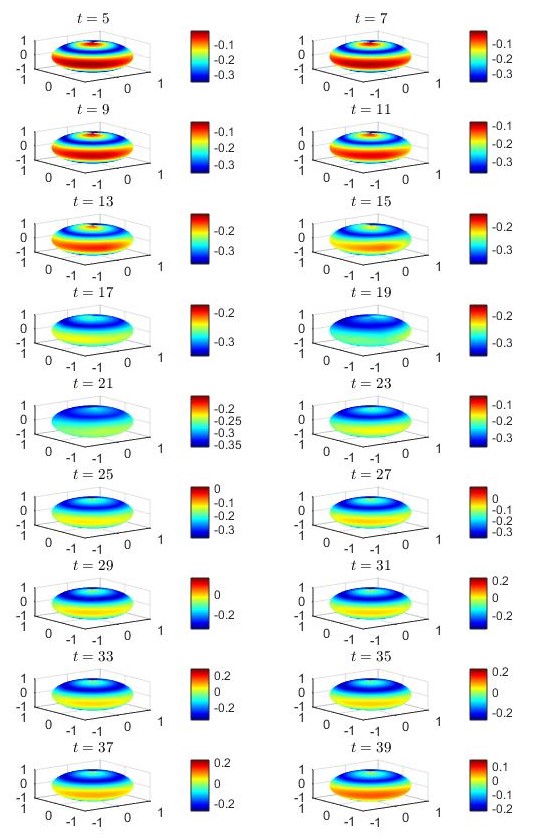}
\caption{Spherical functional 5--fold-cross validation errors associated with response regression predictor. Their spherical functional values are displayed at times \linebreak $t=5, 7, 9, 11, 13, 15, 17, 19, 21, 23, 25, 27, 29, 31, 33, 35, 37, 39$}
\label{f7app}
\end{center}
\end{figure}

\begin{figure}[H]
\begin{center}
\includegraphics[height=0.95\textheight, width=0.8\textwidth]{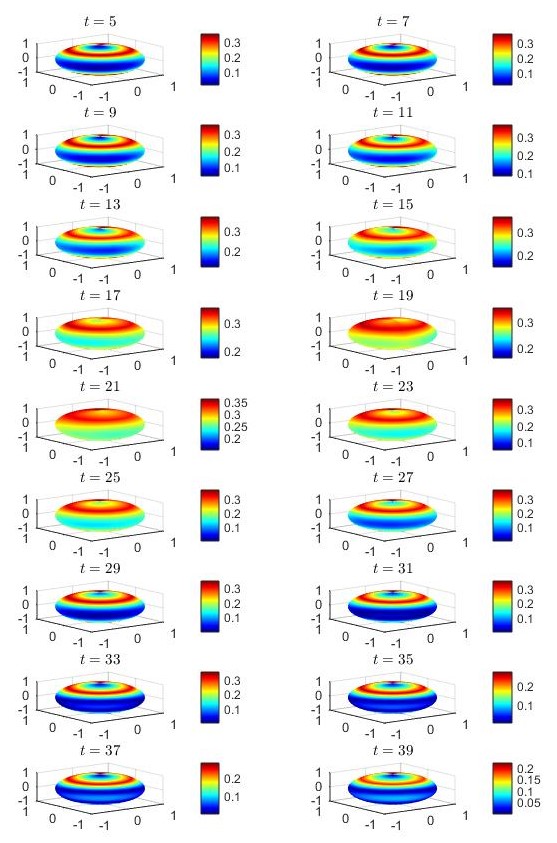}
\caption{Spherical functional 5--fold-cross validation absolute errors associated with response  plug--in regression  predictor. Their spherical functional values are displayed at times \linebreak $t=5, 7, 9, 11, 13, 15, 17, 19, 21, 23, 25, 27, 29, 31, 33, 35, 37, 39$}
\label{f9app}
\end{center}
\end{figure}
The performance of the proposed regression prediction technique is illustrated by the implementation of $5$--fold random cross validation. The spherical functional $5$--fold random cross validation errors obtained by computing the proposed functional regression predictor are displayed in Figure \ref{f7app}. The absolute $5$--fold random cross validation errors associated with  the plug--in functional regression predictor, after minimum contrast estimation of the error term   are also plotted in Figure \ref{f9app}. Note that  a slight difference between regression   and  plug-in regression  performances is observed in the order of magnitude of the modulus of the $5$--fold cross validation errors.

\section{Conclusions}
\label{secconclusions}

This paper opens a new research line within the context of nonlinear multiple functional  regression from manifold   functional data strong--correlated in time. Particularly, the framework of   connected  and compact two--point homogeneous spaces is adopted. The  formulated multiple functional regression model, with functional response, functional regression parameters and time--dependent scalar covariates, goes beyond the  assumptions of weak--dependent, and the Euclidean setting usually adopted in the current literature in functional regression.  The simulation study and real--data application illustrate the interest of the presented approach, allowing the incorporation of time in the covariates,   to represent the  evolution of nonlinear associations between the manifold response  and regressors. In particular, this aspect is crucial  when   changes over time arise modifying  in a substantial way the manifold patterns of functional response and regressors. On the other hand, the linear case addressed in the Supplementary Material (one way FANOVA model in the spatiotemporal spherical context) by projection into a different orthogonal basis, defined from  Jacobi polynomials, allows the prediction of local behaviors in a neighborhood of the pole of the zonal functions considered, which can be of interest in detecting small local changes in the functional response mean in those small areas near the pole.

One of  the most outstanding problems in functional regression is model selection. A  wide variety of statistical tests has been derived for model checking in functional regression. Statistical functional regression model testing techniques  have mainly  been developed for independent and  weak--dependent functional data. See,  for example,  \cite{Dette11};  \cite{Cuesta}; \cite{FebreroBande15}; \cite{Maistre20};  \cite{Patilea16}, in the framework  of functional regression under independent data, and \cite{Alvarez22};  \cite{Constantinou18}; \cite{GMRMFB2024}; \cite{Gorecki18}; \cite{Horman18}; \cite{Horvath2014}; \cite{HorvathReed13};  \cite{Zhang16}, for the case of weak--dependent data in the context of functional time series.

A challenging topic is model checking  under strong dependent functional data, which is the scenario analyzed in the present  paper  (see,  e.g., \cite{Ruiz-MedinaCrujeiras24}, for SRD/LRD model checking in the spectral domain in manifold--supported functional time series). The proposed strategy  to test the suitability of the formulated nonlinear functional regression model from the observed data is to first apply  the results in \cite{Ruiz-MedinaCrujeiras24} for SRD model checking. If the null hypothesis on SRD is not rejected, then one can apply the tools provided in the above cited references, under a weak--dependent data scenario, for  checking additional regression model characteristics.  Otherwise, as given in the simulation study undertaken,  and real--data application, in Sections \ref{secsimulation} and  \ref{secApplic} of the paper,  an empirical analysis, and  data--driven model checking   can be respectively implemented,  when the alternative hypothesis on LRD holds.  In particular, in Section 5, model checking has been performed  in terms of $5$--fold random cross validation. We theoretical address the topic of functional  regression model checking  under  strong--dependent functional data in a subsequent paper.

\subsection*{Acknowledgements}

This work has been supported in part by projects
MCIN/ AEI/PID2022-142900NB-I00,  and  CEX2020-001105-M MCIN/ AEI/10.13039/501100011033.

\end{document}